\documentclass[runningheads,oribibl,]{article}

\newif\ifincludeproofs
\newif\iffinal
\finaltrue

\usepackage{calc}
\usepackage{ifthen}
\usepackage[utf8]{inputenc}
\usepackage{xspace}
\usepackage{amsmath, amssymb, latexsym, amsfonts}
\usepackage[only, lightning, llbracket, rrbracket]{stmaryrd}
\usepackage{enumerate}
\usepackage{multirow}
\usepackage{rotating}
\usepackage{thmtools,thm-restate}

\usepackage[hypertexnames=true,
            plainpages=false,
            linktocpage=false,
            breaklinks=true,
            pdfpagelabels=true,
            pageanchor=true,
            bookmarks=false,
            pdftex,
            final,
            verbose=false,
            debug=false,
            a4paper,
            naturalnames=false,
                        hyperindex=true]{hyperref}

\hypersetup{
         pdftitle={},
         pdfauthor={},
         pdfsubject={},
         pdfkeywords={},
         pdfcreator={pdfLaTeX},
         pdfpagemode=None,
         colorlinks,
         anchorcolor=black,
         citecolor=black,
         filecolor=black,
         linkcolor=black,
         urlcolor=black,
                  pdftex}
\usepackage{graphicx}
\usepackage{tikz}
\usetikzlibrary{matrix,calc,positioning,fit,arrows,decorations.markings,decorations.pathmorphing,backgrounds,automata,shapes.symbols,shapes.arrows,shapes.geometric,shapes.callouts,decorations.pathreplacing,decorations.markings,decorations.pathmorphing,decorations.text}

\usepackage[final]{listings}

\definecolor{darkviolet}{rgb}{0.5,0,0.4}
\definecolor{darkgreen}{rgb}{0,0.4,0.2}
\definecolor{darkblue}{rgb}{0.1,0.1,0.9}
\definecolor{darkgrey}{rgb}{0.5,0.5,0.5}
\definecolor{lightblue}{rgb}{0.4,0.4,1}
\lstdefinestyle{eclipsish}{
    basicstyle=\scriptsize\ttfamily,
    emphstyle=\color{red}\bfseries,
    keywordstyle=\color{darkgreen}\bfseries,
    keywordstyle=[2]\color{darkviolet}\bfseries,
    commentstyle=\color{darkgrey},
    stringstyle=\color{darkblue},
    numberstyle=\color{darkgrey}\ttfamily\tiny,
    emphstyle=\color{red},
        morecomment=[s][\color{lightblue}]{/**}{*/},
     showstringspaces=false,
  numbers=left,
  numbersep=5pt,
  xleftmargin=2.5ex,
  xrightmargin=2.4ex,
  breakindent=3ex,
  breakautoindent,
  numberblanklines=false,
  escapeinside={(*@}{@*)},
  mathescape=true,
}

\lstdefinelanguage{GCD}[]{C}
  {morekeywords=[2]{write, read,dispatch\_s, dispatch\_a,forkjoin,c\_queue, s\_queue, sleep, apply,require,?,!,??,!!},   alsoletter={^},   morekeywords= {def,foreach,block, ^,string,bool, global,  elseif,
   group,select,with,where,in, wait}
  }

\newcommand{\mleq}{\preceq}
\newcommand{\tuple}[1]{\langle#1\rangle\xspace}
\newcommand{\set}[2]{\left\{#1\,\vert\,#2\right\}}
\newcommand{\veve}[1]{\ensuremath{\left| #1 \right|}}
\newcommand{\size}[1]{\veve{#1}}

\newcommand{\Aa}{\ensuremath{\mathcal{A}}\xspace}

\newcommand{\Cc}{\ensuremath{\mathcal{C}}\xspace}
\newcommand{\Dd}{\ensuremath{\mathcal{D}}\xspace}
\newcommand{\Ff}{\ensuremath{\mathcal{F}}\xspace}
\newcommand{\Ll}{\ensuremath{\mathcal{L}}\xspace}
\newcommand{\Pp}{\ensuremath{\mathcal{P}}\xspace}

\newcommand{\Ss}{\ensuremath{\mathcal{S}}\xspace}

\newcommand{\Ts}{\ensuremath{\mathcal{TS}}\xspace}
\newcommand{\Xx}{\ensuremath{\mathcal{X}}\xspace}

\newcommand{\e}{\ensuremath{\varepsilon}\xspace}

\newcommand{\mbar}{\ensuremath{\overline{m}}\xspace}

\newcommand{\DD}{\ensuremath{\mathbb{D}}\xspace}
\newcommand{\NN}{\ensuremath{\mathbb{N}}\xspace}

\newcommand{\sem}[1]{\ensuremath{\llbracket #1 \rrbracket}}

\newcommand{\cfont}[1]{\ensuremath{\mathtt{#1}}\xspace}
\DeclareMathOperator{\dotcup}{\mathaccent\cdot\cup}
\DeclareMathOperator{\bigdotcup}{\ooalign{$\bigcup$\cr\hfill$\cdot$\hfill}}

\newcommand{\main}{\ensuremath{\textit{main}}\xspace}
\DeclareMathOperator{\incr}{\cfont{incr}}
\DeclareMathOperator{\decr}{\cfont{decr}}
\DeclareMathOperator{\zerotest}{\cfont{is\_zero}}
\DeclareMathOperator{\ack}{\cfont{ack}}

\newcommand{\CQID}{CQID\xspace}
\newcommand{\SQID}{SQID\xspace}
\newcommand{\QID}{QID\xspace}

\newcommand{\lts}{\textsc{Lts}\xspace}
\newcommand{\pds}{\textsc{Pds}\xspace}
\newcommand{\mpds}{\textsc{Mpds}\xspace}
\newcommand{\petri}{\textsc{Pn}\xspace}
\newcommand{\pn}{\textsc{Pn}\xspace}
\renewcommand{\gcd}{\textsc{Gcd}\xspace}
\newcommand{\qdas}{\textsc{Qdas}\xspace}

\newcommand{\eqdas}{\textsc{eQdas}\xspace}
\newcommand{\twocs}{2\textsc{Cs}\xspace}

\newcommand{\fifo}{fifo\xspace}
\newcommand{\dexpspace}{\textsc{ExpSpace}\xspace}
\newcommand{\dexpspacecomplete}{\textsc{ExpSpace-C}\xspace}
\newcommand{\pspace}{\textsc{PSpace}\xspace}
\newcommand{\pspacecomplete}{\textsc{PSpace-C}\xspace}
\newcommand{\ptime}{\textsc{PTime}\xspace}
\newcommand{\ptimecomplete}{\textsc{PTime-C}\xspace}

\newcommand{\dexptime}{\textsc{ExpTime}\xspace}
\newcommand{\dexptimecomplete}{\textsc{ExpTime-C}\xspace}

\newcommand{\QDAS}{\ensuremath{\tuple{\CQID, \emptyset,\allowbreak \Gamma,
\allowbreak \mathtt{main},\allowbreak \Xx,\allowbreak \Sigma,\allowbreak
(\Ts_\gamma)_{\gamma\in\Gamma}}}}
\DeclareMathOperator{\Reach}{\textit{Reach}}

\newcommand{\Cover}{\ensuremath{\textit{Cover}}}
\newcommand{\Reachloss}{\Cover}

\newcommand{\SSigma}{\ensuremath{\smash{\widetilde{\Sigma}}}}

\newcommand{\Graph}{\ensuremath{G}}
\newcommand{\Data}{\ensuremath{\vec{d}}}
\newcommand{\DData}{\ensuremath{\widehat{\vec{d}}}}

\newcommand{\queue}{\ensuremath{\textit{queue}}}
\newcommand{\state}{\ensuremath{\textit{state}}}
\newcommand{\head}{\ensuremath{\textit{head}}}
\newcommand{\tail}{\ensuremath{\textit{tail}}}
\newcommand{\enqueue}{\ensuremath{\textsf{enqueue}}}
\newcommand{\dequeue}{\ensuremath{\textsf{dequeue}}}
\newcommand{\step}{\ensuremath{\textsf{step}}}
\newcommand{\letwait}{\ensuremath{\textsf{letwait}}}

\newcommand{\ctg}{\ensuremath{\textsc{Ctg}}\xspace}
\newcommand{\disps}{\ensuremath{\cfont{dispatch_s}}}
\newcommand{\dispa}{\ensuremath{\cfont{dispatch_a}}}
\newcommand{\forkjoin}{\ensuremath{\cfont{forkjoin}}\xspace}
\newcommand{\guardson}[1]{\ensuremath{\mathsf{guards}\left(#1\right)}}
\newcommand{\valof}[1]{\ensuremath{\mathsf{vals}\left(#1\right)}}
\newcommand{\assignon}[1]{\ensuremath{\mathsf{assign}\left(#1\right)}}
\newcommand{\Parikh}{\ensuremath{\mathsf{Parikh}}}
\newcommand{\mystate}[2]{\ensuremath{s^{#1}_{\mathtt{#2}}}}

\newcommand{\push}{\ensuremath{\cfont{push}}}
\newcommand{\pop}{\ensuremath{\cfont{pop}}}
\newcommand{\emptystack}{\ensuremath{\cfont{empty?}}}

\newcommand{\ko}{\ensuremath{\lightning}}

\newcommand{\vstrut}[1]{\rule{0pt}{#1}}

\newcommand{\inputproof}[1]{\ifincludeproofs\else\fi}

\newenvironment{myitemize}{\begin{list}{\labelitemi}{\setlength{\topsep}{4pt}\setlength{\partopsep}{0pt}
\setlength{\itemsep}{0pt}
\setlength{\itemindent}{0ex}
\setlength{\listparindent}{0ex}
\setlength{\leftmargin}{4ex}\setlength{\labelwidth}{2ex}
}}
{\end{list}}

\tikzstyle{state}=[draw,ellipse,inner sep=1pt,minimum size=2.5ex,font=\small]
\tikzstyle{box}=[state,rectangle,minimum size=2.3ex]
\tikzstyle{lop}=[->,looseness=5,out=30,in=-30,inner sep=1pt,shorten >=1pt]
\tikzstyle{lab}=[font=\tiny]
\tikzstyle{na}=[baseline=-0.5ex]

\def\eor{\ifmmode\squareforqed\else{\unskip\nobreak\hfil
\penalty50\hskip1em\null\nobreak\hfil$\dashv$
\parfillskip=0pt\finalhyphendemerits=0\endgraf}\fi}

\newtheorem{theorem}{Theorem}{}
\newtheorem{proposition}{Proposition}{}
\newtheorem{lemma}{Lemma}{}
{}
\newtheorem{example}{Example}{}
\newenvironment{proof}{\noindent{\it Proof.\hspace*{.5cm}}}{}
\newcommand{\qed}{\hfill$\Box$}

\newcommand\fixme[1]{}
\newcommand\detail[1]{}
\newcommand\change[1]{}
\newcommand\afaire[1]{}

\begin{document}
\title{Queue-Dispatch Asynchronous Systems}
\author{Gilles Geeraerts
\and Alexander Heußner
\and Jean-François Raskin\\[1.5ex]Université Libre de Bruxelles -- Belgium}
\maketitle

\begin{abstract}
  To make the development of efficient multi-core applications easier,
  libraries, such as Grand Central Dispatch, have been proposed. When
  using such a library, the programmer writes so-called {\em blocks},
  which are chunks of codes, and dispatches them, using {\em
    synchronous} or {\em asynchronous} calls, to several types of
  waiting queues. A scheduler is then responsible for dispatching
  those blocks on the available cores.  Blocks can synchronize via a
  global memory.  In this paper, we propose Queue-Dispatch
  Asynchronous Systems as a mathematical model that faithfully
  formalizes the synchronization mechanisms and the behavior of the
  scheduler in those systems.  We study in detail their relationships
  to classical formalisms such as pushdown systems, Petri nets, fifo
  systems, and counter systems.  Our main technical contributions are
  precise worst-case complexity results for the Parikh coverability
  problem and the termination question for several subclasses of our
  model. We give an outlook on
  extending our model towards verifying input-parametrized fork-join
  behaviour with the help of abstractions.
\end{abstract}

\section{Introduction\label{sec:introduction}}
The computing power delivered by computers has followed an exponential growing
rate the last decades. One of the main reasons was the steady increase of the CPU
clock rates. This growth, however, has come to an end a few years ago, because
further increasing the clock rate would incur major engineering challenges
related to power dissipations.  In order to overcome this and meet the continuous
need for more computing power, multi-core CPU's have been introduced and are now
ubiquitous. However, in order to harness the power of multiple cores, software
applications need to be fundamentally modified and the programmers now have to
write programs with parallelism in mind. But writing parallel programs is a
notoriously difficult and error prone task. Also, writing {\em efficient} and
{\em portable} parallel code for multi-core platforms is difficult, as the number
of available cores will vary greatly from one platform to another, and might also
depend on the current load, the energy management policy, and so forth.

In order to alleviate the task of the programmer, several high level programming
interfaces have been proposed, and are now available on several operating
systems. A popular example is {\em Grand Central Dispatch}, \gcd for
short, a technology that is present in Mac\,OS X (since 10.6), iOS (since version
4), and FreeBSD.
In \gcd, the programmer writes so-called {\em blocks} which are
chunks of codes, and send them to {\em queues}, together with several dependency
constraints between those blocks (for instance, one block cannot start before the
previous one in the queue has finished). The scheduler is then responsible for
dispatching those blocks on the available cores, through a thread pool that the
scheduler manages (thereby avoiding the explicit and costly creation/destruction
of threads by the programmer that is in addition extremely error-prone).

So far, to the best of our knowledge, no formal model has been proposed for
systems relying on \gcd or similar technologies, making those programs {\it de facto} out of reach of current verification methods and tools.  This is particularly unfortunate as the control structure of such programs is rich and may exhibit complex behaviors. Indeed, the state-space of such programs is infinite even when types of variables are abstracted to finite domains of values. This is not surprising as asynchronous calls and recursive synchronous calls can send an unbounded number of blocks to queues. Also, those programs are, as any parallel program, subject to concurrency bugs that are difficult to detect using testing only.

\begin{table}[t!]
  \centering
  \scalebox{0.8}{
  \begin{tabular}{ll|ccc}
    \multicolumn{2}{c|}{Parikh coverability}& \multicolumn{3}{c}{queue types}\\[1.5ex]
    \multicolumn{2}{c|}{}& concurrent & serial & both\\
    \cline{1-5}
    \multirow{3}{*}{\begin{sideways}{dispatch\ \ }\end{sideways}}      \hspace{1ex}\vstrut{3ex}
    &synchr.  & \dexptimecomplete & \ \ \ \pspacecomplete\ \ \  & \dexptimecomplete\\[1.5ex]
    &asynchr.\ \ & \ \ \dexpspacecomplete & \ko & (\ko)\\[1.5ex]
    &both  & \ko & (\ko) & (\ko)\\
  \end{tabular}}
                                            \\[1.5ex]\mbox{ }

  \caption{\qdas Verification Problems (\ko: ``undecidable'',
  parentheses: directly derivable)\label{tab:decidability}}
\end{table}

\paragraph{{\bf Contributions}}
In this paper, we introduce {\em Queue-Dispatch Asynchronous Systems},
\qdas for short, as a formal model for programs written using
libraries such as \gcd.  Our model is composed of {\em blocks}, that
are finite transition systems with finite data-domain variables that
can do {\em asynchronous} (non-blocking) and {\em synchronous}
(blocking) calls to other blocks (possibly recursively). However, a
call does not immediately trigger the execution of the callee: the
block  is inserted into a queue that can be either {\em
  concurrent} or {\em serial}. In concurrent queues, several blocks
can be taken from the queue and executed in parallel, while in serial
queues, a block can be dequeued only if the previous block in the
queue has completed its execution. Queues are maintained with a \fifo
policy. To formalize configurations of such systems, our formal
semantics relies on {\em call task graph}, \ctg for short, in which
nodes model tasks that are either in queues or executing, and edges
model dependencies between tasks and within queues.

We then study the decidability border for the {\em Parikh coverability
problem} and the \emph{termination problem}
on several subclasses of \qdas. Our results are summarized in Table\,\ref{tab:decidability}. The {\em Parikh image} of a
\ctg is an abstraction that counts for each type and state of blocks
the number of occurrences in the \ctg and the {\em Parikh
  coverability} problem asks for the reachability of a \ctg that
contains at least a given number of blocks of each type that are in a given set
of states. Not
surprisingly, this problem is undecidable for \qdas, but we identify
several subclasses for which the problem is decidable. For those
decidable cases, we characterize the exact complexity of the
problem.

The main positive decidability results with precise complexity are as
follows:  First, we show that \qdas with {\em only} synchronous calls
are essentially equivalent to pushdown systems with finite domain
data-variables, and we show that the Parikh coverability problem is
\dexptimecomplete for synchronous concurrent \qdas
(Theorem\,\ref{thm:syncqdas}). Second, for synchronous \qdas with only
serial queues, the problem is \pspacecomplete
(Theorem\,\ref{thm:serialsyncqdas}). Third, we show that \qdas with
{\em only} asynchronous calls and {\em only} concurrent queues are
essentially equivalent to lossy Petri nets and show that the
Parikh coverability problem is \dexpspacecomplete for that class
(Theorem\,\ref{thm:concasyncqdas}). This decidability border is precise
as we show that if we allow either $(i)$ asynchronous calls with
synchronous queues, or $(ii)$ synchronous and asynchronous calls with
concurrent queues, then the Parikh coverability problem becomes
undecidable (Theorem\,\ref{the:async-seri-undec} and
Theorem\,\ref{thm:concqdasundec}). The previous proof's ideas allow
to derive similar results for
termination wrt. the subclasses of \qdas.
The \emph{termination} problem asks given a \qdas whether all its
executions are finite.

We enhance up our results by presenting an extension of \qdas with an
explicit fork/join construct that, in addition, is parametrized by the
input. As Parikh coverability and termination lifted to this setting
are undecidable, we propose two over-approximations that allow for
solutions in practice.

\medskip
\noindent {\it Remark}: Due to the lack of space, detailed formal
proofs are deferred to the appendix.

\paragraph{{\bf Related Works}}
The basic model checking result for asynchronous programs is the
\dexpspace-hardness for the control-state reachability problem
obtained by making formal a link with \emph{multi-set pushdown
systems} (\mpds). The underlying two basic ideas are : $(i)$ to
untangle the call stack and the storage of pending asynchronous calls
by imposing that the next call in a serialized execution-equivalent
program is only processed when the call stack is empty; and $(ii)$ to
only count the number of pending calls for each block while the call
stack is non-empty. The original reduction in \cite{sen-k-2006-300-a}
is based on Parikh's theorem and derives the lower bound from a Petri
net reachability problem~\cite{esparza-j-1998-374-a}.  A Parikh-less
reduction was presented in~\cite{jhala-r-2007-339-a} that relied on
the convergence of an over- and under-approximation derived from
interprocedural dataflow analysis.

The close relation between asynchronous programs and Petri nets can also be used
to prove additional decidability results for
liveness questions~\cite{ganty-p-2009-102-a,ganty-p-2010--a}. The following results are based on a (polynomial-time)
reduction of asynchronous systems to an ``equivalent'' Petri net or extension
thereof: \emph{fair} termination (i.e., testing whether each dispatched call
terminates) is complete in \dexpspace, the
boundedness question is decidable in \dexpspace (i.e., asking whether we can
bound the number of pending calls), fair non-starvation (i.e., asking, when
assuming fairness on runs, whether every pending call is eventually dispatched)
is decidable.
The authors also consider extensions of asynchronous programs with cancellation
(i.e., an additional operation removing all pending instances of a block) and
testing whether there is \emph{no} pending instance of a given block. In the
first case, they show reduction to the model to Petri nets with transfer arcs or
reset arcs, in the second case they show reduction to Petri nets with one
inhibitor arc.
Multi-set pushdown automata are subsumed by \emph{well-structured
  transition systems with auxiliary storage} and inherit their
decidability results presented in
\cite{chadha-r-2007-136-a,chadha-r-2009-4169-a}.
Analogously, one can show that
termination, control-state maintainability, and simulation with
respect to finite state systems are decidable for asynchronous
programs.

All the models considered in the aforementioned publications do not
consider causality constraints on the sequence of asynchronous
dispatch calls, as would be necessary to model the \fifo policies of
\gcd. However, this is possible with \qdas. A more detailed look on
the differences between the model of~\cite{ganty-p-2010--a} and the
(\fifo-less) subclass of asynchronous serial \qdas is presented in
Section\,\ref{sec:parikh-cover-probl}.

A series of parallel programming libraries and techniques is
formalized in~\cite{bouajjani-a-2012--a} with the help of
\emph{recursively parallel programs}. These allow to model fork/join
based parallel computations based on a reduction to recursive vector
addition systems with states. With respect to \qdas and asynchronous
programming, recursively parallel programs only cover the classical
asynchronous models presented above and not the advanced scheduling
strategies for different queues that introduce more sophisticated
behaviours.

\section{Preliminaries}\label{sec:prelims}
\emph{\bf Grand Central Dispatch} (\gcd) is a technology developed by
Apple \cite{apple--2010--a,apple--2011--a} that is publicly available
at \url{http://libdispatch.macosforge.org/} under a free license. \gcd
is the main inspiration for the formal model of queue-dispatch
asynchronous systems. In the following, we often present our examples
as pseudo code using a syntax inspired by \gcd.  In the \gcd
framework, the programmer has to organize his code into
\emph{blocks}. During the execution of a \gcd program, one or several
\emph{tasks} run in parallel, each executing a given block (initially,
only the \texttt{main} block is running). Tasks can call (or
\emph{dispatch} in the \gcd vocabulary) other blocks, either
\emph{synchronously} (the call is blocking), or \emph{asynchronously}
(the call is not blocking). A \emph{dispatch} consists in inserting
the block into a \fifo \emph{queue}. In our examples, we use the
keywords $\dispa$ and $\disps$ to refer to asynchronous and
synchronous dispatches respectively.  At any time, the scheduler can
decide to \emph{dequeue} blocks from the queues and to assign them to
tasks for execution. All queues ensure that the blocks are dequeued in
\fifo order, however the actual scheduling policy depends on the type
of queue. \gcd supports two types of queues: \emph{concurrent queues}
allow several tasks from the same queue to run in parallel, whereas
\emph{serial queues} guarantee that \emph{at most one} task from this
queue is running. In our examples, concurrent (or serial) queues are
declared as global variables of type \texttt{c\_queue}
(\texttt{s\_queue}). In addition, all blocks have access to the same
set of \emph{global variables} (in this work, we assume that the
variables range over finite domains).

\begin{example}\label{ex:gcd-matrixmult}
  Let us consider the pseudo code in Fig.~\ref{fig:example-gcd} that
  computes the product of two integer matrices $\texttt{matrix1}$ and
  $\texttt{matrix2}$ of constant size ($\texttt{l,m,n}$)
  in a matrix $\texttt{matrix}$. The \texttt{main}
  task forks a series of \texttt{one\_cell} blocks. Each
  \texttt{one\_cell} computes the value of a single cell of the
  result. The parallelism is achieved via the \gcd scheduler, thanks
  to \emph{asynchronous dispatches} on the \emph{concurrent} queue
  \texttt{workqueue}. Asynchronous dispatches are needed to make sure
  that \texttt{main} is not blocked after each dispatch, and a
  concurrent queue allow all the \texttt{one\_cell} block to run in
  parallel.  The variable \texttt{count} is incremented each time the
  computation of a cell is finished and acts as a semaphore for the
  $\texttt{main}$ block, to ensure that \texttt{matrix} contains the
  final result. As only reading and writing to a variable are atomic,
  we need to guarantee exclusive access of two consecutive operations
  on \texttt{count} (line \ref{lst:count}). This is achieved by a
  dedicated block \texttt{increase} that is dispatched to the
  \emph{serial} queue \texttt{semaphore}. As only \texttt{increase}
  blocks can increase \texttt{count}, this queue implicitly locks the
  access to the variable.  Moreover, the synchronous dispatch in line
  \ref{lst:syncdisp}  guarantees that a block terminates only
  after it has increased \texttt{count}.
      \end{example}

\begin{figure}[t!]
  \centering
  \begin{lstlisting}
global int const l,m,n  (*@\label{lst:const}@*)
global int[l][m] matrix1, int[m][n] matrix2, int[l][n] matrix
global c_queue workqueue, s_queue semaphore, int count (*@\mbox{}\\[0.5ex]@*)
block increase():
  count = count + 1 (*@\label{lst:count}\\[0.5ex]@*)
block one_cell(int i, int j):
  for k in range(m):
    matrix[i][j]+= matrix1[i][k] * matrix2[k][j]
  dispatch_s(semaphore,increase()) (*@\label{lst:syncdisp}\\[0.5ex]@*)
def main():
  // read input matrix1, matrix2
  count = 0
  for i in range(l):  (*@\label{lst:fork-start}@*)
    for j in range(n):
      dispatch_a(workqueue,one_cell(i,j))
  wait(count = l*n)  (*@\label{lst:join}@*)
  // print the result
  \end{lstlisting}
  \caption{\gcd(-like) program for parallel matrix multiplication}
\label{fig:example-gcd}
\end{figure}

\paragraph{\bf Basic Notations:} Given a set $S$, let $\size{S}$ denote its
cardinality. For an $I$-indexed family of sets $(S_i)_{i \in I}$, we
write elements of $\prod_{i \in I} S_i$ in bold face, i.e.,
$\vec{s}\in\prod_{i\in I} S_i$. The $i$-component of $\vec{s}$ is
written $s_i \in S_i$, and we identify $\vec{s}$ with the indexed
family of elements $(s_i)_{i \in I}$.
We use $\dotcup$ to denote the disjoint union of sets.
An \emph{alphabet} $\Sigma$ is a finite set of \emph{letters}. We
write $\Sigma^*$ for the set of all \emph{finite words}, over $\Sigma$
and denote the empty word by $\varepsilon$. The concatenation of two
words $w,w'$ is represented by $w\cdot w'$. For a letter
$\sigma\in\Sigma$ and a word $w\in\Sigma^*$, let $\veve{w}_\sigma$ be
the number of occurrences of $\sigma$ in $w$.
We use standard complexity classes, e.g., polynomial time
(\ptime) or deterministic exponential time (\dexptime), and mark
completeness by appending ``-C'' (\pspacecomplete).

Let $\DD$ be a finite \emph{\bf data domain}
with an \emph{initial element} $d_0\in\DD$, and let $\Xx$ be a finite
set of variables ranging over $\Dd$. A \emph{valuation} of the
variables in $\Xx$ is a function $\mathbf{d}: \Xx\rightarrow \DD$.  An
\emph{atom} is an expression of the form $x=d$ or $x\neq d$, where
$x\in\Xx$ and $d\in \DD$. A \emph{guard} if a finite conjunction of
atoms. An assignment is an expression of the form $x\gets v$, where
$x\in \Xx$ and $v\in \DD$. Let $\guardson{\Xx}$, $\assignon{\Xx}$ and
$\valof{\Xx}$ denote respectively the sets of all guards, assignments
and valuations over variables from $\Xx$. Guards, atoms and valuations
have their usual semantics: for all valuations $\Data$ of $\Xx$ and all
$g\in\guardson{\Xx}$, we write $\Data\models g$ iff $\Data$ satisfies
$g$.

\smallskip\noindent A \emph{\bf pushdown system with data} is a
pushdown system (see~\cite{bouajjani-a-1997-135-a} for details)
equipped with a finite set of variables $\Xx$ over a
finite domain $\DD$. A configuration of a \pds with data is a pair
$(s,w,\Data)$ where $s$ is a control state, $w$ is the stack content,
and $\Data$ is a valuation of the variables

\begin{restatable}{proposition}{propreachpdsdata}\label{prop:reach_pds_data}
  The reachability problem is \dexptimecomplete for \pds with data.
\end{restatable}

\smallskip\noindent A \emph{\bf Petri net} (\petri) is a tuple
$N=\langle P, T, m_0 \rangle$ where $P$ is a finite set of places, a
\emph{marking} of the places is function $m:P\rightarrow \NN$ that
associates, to each place $p\in P$ a number $m(p)$ of tokens, $T$ is
finite set of transitions, each transition $t\in T$ is a pair $(I_t,
O_t)$ where
          $I_t: P\rightarrow \{0,1\}$ and $O_t:
P\rightarrow \{0,1\}$ are respectively the \emph{input} and
\emph{output functions} of $t$, and $m_0$ is the \emph{initial
  marking}. Given two markings $m_1$ and $m_2$, we let $m_1\mleq m_2$
iff $m_1(p)\leq m_2(p)$ for all $p\in P$. Given a marking $m$, a
transition $t=(I_t, O_t)$ is \emph{enabled} in $m$ iff $m(p)\geq
I_t(p)$ for all $p\in P$. When $t$ is enabled in $m$, one can
\emph{fire} the transition $t$ in $m$, which produces a new marking
$m'$ s.t. $m'(p)=m(p)-I_t(p)+O_t(p)$ for all $p$. This is denoted
\smash{$m\xrightarrow{t} m'$}, or simply $m\rightarrow m'$ when the
transition identity is irrelevant.  A \emph{run} is a finite sequence
$m_0m_1\ldots m_n$ s.t. for all $1\leq i\leq n$: $m_{i-1}\rightarrow
m_i$.  For a \petri $N$, we denote by $\Reach(N)$ (resp. $\Cover(N)$)
the \emph{reachability (coverability) set} of $N$, i.e. the set of all
markings $m$ s.t. there exists a run $m_0m_1\ldots m_n$ of $N$ with
$m=m_n$ ($m\mleq m_n$). The \emph{coverability problem} asks, given a
\petri $N$ and a marking $m$, whether $m\in\Cover(N)$. It is
\dexpspace-complete~\cite{esparza-j-1998-374-a}. The termination
problem, i.e., whether all executions of the Petri net
are finite, is decidable in \dexpspacecomplete~\cite{lipton,rackoff}.

\section{Queue-dispatch asynchronous systems}

\paragraph{\bf Syntax:} We now define our formal model for
queue-dispatch asynchronous systems.
Let $\DD$ be a finite data domain containing an \emph{initial value}
$d_0$. A \emph{queue-dispatch asynchronous system} (\qdas) $\Aa$ is
a tuple $\tuple{ \CQID, \SQID, \Gamma, main, \Xx, \Sigma,
  (\Ts_\gamma)_{\gamma\in\Gamma}}$ where:
\begin{myitemize}
  \item $\CQID$ and $\SQID$ are respectively sets of
    \emph{(c)oncurrent} and \emph{(s)erial queues};
  \item $\Gamma$ is the finite set of \emph{blocks} and
    $main\in\Gamma$ the \emph{initial block}.  Each block
    $\gamma\in\Gamma$ is a tuple $\langle S_\gamma, s^0_\gamma,
    f_\gamma, \Sigma, \Delta_\gamma \rangle$ where $\langle S_\gamma,
    s^0_\gamma, \Sigma, \Delta_\gamma\rangle$ is an \lts and
    $f_\gamma\in S$ a distinct final state;
  \item $\Xx$ is a finite set of $\DD$-valued variables;
  \item $\Sigma$ is the set of \emph{actions}, with
    $\Sigma=(\{\disps,\dispa\}\times (\CQID\cup\SQID)$\linebreak[1] $\times
    \Gamma\setminus\{main\}) \cup \guardson{\Xx} \cup \assignon{\Xx}$.
\end{myitemize}

We assume that $\SQID, \CQID, \Gamma, \Xx$, and all $S_\gamma$ for
$\gamma\in\Gamma$ are disjoint from each other. Let
$S=\bigdotcup_{\gamma\in\Gamma}S_\gamma$,
$F=\bigdotcup_{\gamma\in\Gamma} \{f_\gamma\}$,
$\Delta=\bigdotcup_{\gamma\in\Gamma}\Delta_\gamma$, and
$\QID=\SQID\dotcup\CQID\dotcup\{\imath\}$ (where $\imath\notin\SQID\cup\CQID$).
We further assume that $\e\notin\Sigma$.

\paragraph{\bf Call-task graphs:}
We formalize the semantics of \qdas using the notion of
\emph{call-task graph} (\ctg) to describe the system's global
configurations.

A configuration of a \qdas (see Fig.\,\ref{fig:example-ctg} for an
example) contains a set of running tasks, represented by \emph{task
  vertices} (depicted by round nodes), a set of called but unscheduled
blocks, represented by \emph{call vertices} (square nodes). Call
vertices are held by queues, and the linear order of each queue is
represented by \emph{queue edges} (solid edges). Synchronous calls add
an additional dependency (the caller is waiting for the termination of
the callee) that is represented by a \emph{wait edge} (dashed edges)
between the caller and the callee. Wait edges are also inserted
between the head of a \emph{serial} queue and the running task that
has been extracted from this queue (if it exists) to indicate that the
task has to terminate before a new block can be dequeued.  Note that
only vertices without outgoing edges can execute a computation step,
the others are currently blocked.  Each node $v$ is labeled by a block
$\lambda(v)$, an by the identifier $\queue(v)$ of the queue that
contains it (for call vertices) or that contained it (for task
vertices). Task vertices are labeled by their current state
$\state(v)$ (for convenience, we also label call vertices by the
initial state of their respective blocks -- not shown in the figure).

\vspace{-1ex}
\begin{example}
  The \ctg in Fig.\,\ref{fig:example-ctg} depicts a configuration of a
  \qdas with two queues. Queue $q_2$ is serial (note the outgoing wait
  edge to the running task) and contains $\gamma_2\gamma_2\gamma_2$,
  and $q_1$ is parallel with content $\gamma_1\gamma_2$. There are 4
  active tasks, two of them (\texttt{main} and the task running
  $\gamma_1$) are blocked. The task running $\gamma_3$ has been
  dequeued from $q_2$ and is currently at location $s$. \eor
\end{example}
\vspace{-1ex}

\begin{figure}[t]
    \centering
        \begin{tikzpicture}
  [zstd/.style={state,font=\small,inner sep=1pt,minimum size=15pt},
  zstd2/.style={zstd,rectangle},
  lab/.style={font=\tiny,inner sep=1pt},
  anchor=west]

  \draw (0,0) node[zstd2] (00) {$\gamma_1$};
    \draw (00.south east) node[lab,anchor=south west] {$q_1$};
  \draw (00)+(1,0) node[zstd2] (01) {$\gamma_2$};
    \draw (01.south east) node[lab,anchor=south west] {$q_1$};
  \draw (00) edge[->] (01);

  \draw (4,0) node[zstd2] (10) {$\gamma_2$};
    \draw (10.south east) node[lab,anchor=south west] {$q_2$};
  \draw (10)+(1,0) node[zstd2] (11) {$\gamma_3$};
    \draw (11.south east) node[lab,anchor=south west] {$q_2$};
  \draw (11)+(1,0) node[zstd2] (12) {$\gamma_2$};
    \draw (12.south east) node[lab,anchor=south west] {$q_2$};
  \foreach \x / \y in {10/11,11/12}
    \draw (\x) edge[->] (\y);

  \draw (12)+(1.5,0) node[zstd] (13) {$\gamma_3$};
    \draw (13.south east) node[lab,anchor=north west] {$q_2$};
    \draw (13.north east) node[lab,anchor=south west] {$s$};
  \draw[dashed] (12) edge[->] (13);

  \draw (0.5,1) node[zstd] (main) {\main};
    \draw (main.south east) node[lab,anchor=north west] {$\imath$};
    \draw (main.north east) node[lab,anchor=south west] {$s$};
  \draw (main)+(1,0) node[zstd] (one) {$\gamma_1$};
    \draw (one.south east) node[lab,anchor=north west] {$q_1$};
    \draw (one.north east) node[lab,anchor=south west] {$s'$};
  \draw[dashed]
        (main) edge[->] (one)
        (one) edge[->,out=0,in=90] (11);

  \draw (2.5,0.1) node[zstd] (two) {$\gamma_2$};
    \draw (two.south east) node[lab,anchor=north west] {$q_1$};
    \draw (two.north east) node[lab,anchor=south west] {$s''$};

    \begin{pgfonlayer}{background}
      \path (01)+(0.5,0)  coordinate (dummy);
      \draw node[fit=(00) (01) (dummy),fill=gray!20,inner sep=4pt] (c1) {};
      \draw (c1.north west)
        node[anchor=south west,inner sep=2pt,font=\scriptsize]
        {queue $q_1$};
      \path (12)+(0.5,0)  coordinate (dummy);
      \draw node[fit=(10) (12) (dummy),fill=gray!20] (c1) {};
      \draw (c1.north west)
        node[anchor=south west,inner sep=2pt,font=\scriptsize]
        {queue $q_2$};
    \end{pgfonlayer}

  \draw[overlay]
  (13) +(0.5,0.8) node[draw,
  rectangle callout,callout absolute pointer={(13.east)+(0.1,0)},font=\scriptsize,
    fill=gray!5,text width=1.8cm, align=center]
  {$\lambda(v)=\gamma_3$ $\state(v)=s$ $\queue(v)=q_2$};
\end{tikzpicture}
\caption{\ctg for a \qdas with a concurrent queue $q_1$ and  a serial queue
    $q_2$\label{fig:example-ctg}}
    \vspace{-3ex}
  \end{figure}
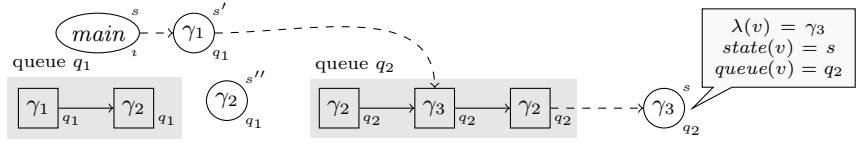

  Formally, given a \qdas $\Aa=\tuple{ \CQID, \SQID, \Gamma, main,
    \Xx, \Sigma, (\Ts_\gamma)_{\gamma\in\Gamma}}$, a \emph{call-task
    graph} over $\Aa$ is a tuple
  $\Graph_\Aa=\tuple{V,E,\lambda,\queue,\state}$ where: $V=V_C\dotcup
  V_T$ is a finite set of \emph{vertices}, partitioned into a set
  $V_C$ of \emph{call vertices} and a set $V_T$ of \emph{task
    vertices}; $E\subseteq V\times V$ is a set of \emph{edges};
  $\lambda: V\rightarrow \Gamma$ labels each vertex by a block;
  $\queue: V\rightarrow \QID\cup\{\imath\}$ associates each vertex to
  a queue  identifier
  (or $\imath$); and $\state: V\rightarrow S$ associates each vertex
  to a \lts state.
  For each $q\in\QID$, let $V_q=\set{v\in V}{\queue(v)=q}$. The set
  $E$ is partitioned into the set $E_W$ of \emph{wait edges} and the
  set $E_Q=\bigdotcup_{q\in\QID} E_q$ of \emph{queue edges} where, for
  each $q\in\QID$, $E_q=E\cap (V_q\times V_q)$.

A \ctg is \emph{empty} iff $V=\emptyset$. The \emph{Parikh image}
$\Parikh(\Graph)$ of a \ctg $\Graph$ of $\Aa$ is a function
$f:S\rightarrow \NN$, s.t. for all $s\in S$: $f(s)=|\set{v\in
  V}{\state(v)=s}|$. Given two Parikh images $\Parikh(G)$ and
$\Parikh(G')$, we let $\Parikh(G)\mleq\Parikh(G')$ iff for all $s\in
S$: $\Parikh(G)(s)\leq\Parikh(G')(s)$.
A \emph{path} (of length $n$) in $\Graph_\Aa$ is a sequence of
vertices $v_0,v_1,\ldots, v_n$ s.t. for all $1\leq i\leq n$:
$(v_{i-1},v_i)\in E$. Such a path is \emph{simple} iff $v_i\neq v_j$
for all $1\leq i<j\leq n$. The \emph{restriction} of $\Graph_\Aa$ to
$V'\subseteq V$ is the \ctg $\Graph_\Aa'=\tuple{V', E',\lambda',
  \queue', \state'}$, where $E'=E\cap (V'\times V')$, and $\lambda'$,
$\queue'$ and $\state'$ are respectively the restrictions of
$\lambda$, $\queue$ and $\state$ to $V'$.

In the rest of the paper, we  assume that all the \ctg we
consider are \emph{well-formed}, i.e., they fulfill the following
requirements:
\begin{enumerate}
  \item For each $v\in V_T$: $\state(v)\in S_{\lambda(v)}$
    where $S_{\lambda(v)}$ are the states of $\Ts_{\lambda(v)}$.
\item Each \emph{call} vertex has at most one outgoing (queue or wait)
  edge, at most one incoming \emph{wait} edge, and at most one
  incoming \emph{queue} edge. Each \emph{task} vertex has at most one
  outgoing, and at most one incoming \emph{wait} edge.
\item For each $q\in \QID$, the restriction of $\Graph_\Aa$ to $V_q$
  is either empty or contains one and only one simple path of length
  $|V_q|-1$. Intuitively, this ensures the well-formedness of the
  queues.
\item For each $q\in \SQID$, there is at most one task vertex $v$
  s.t. $\queue(v)=q$. This ensures that queues in $\SQID$ indeed
  force the serial execution of its members.
\end{enumerate}

For convenience, we also introduce the following notations. Let
$\Graph_\Aa$ be a \ctg, and let $q$ be a queue identifier of
$\Aa$. Then, $\head(q,\Graph_\Aa)$ and $\tail(q,\Graph_\Aa)$ denote
respectively the head and the tail of $q$ in the configuration
described by $\Graph_\Aa$, that is, $\head(q,\Graph_\Aa)$ is the call
vertex $v\in V_q$ that has no incoming queue edge, or $\bot$, if such
a vertex does not exist; and $\head(q,\Graph_\Aa)$ is the call vertex
$v\in V_q$ that has no outgoing \emph{queue} edge (but possibly an
outgoing \emph{wait} edge), or $\bot$, if such a vertex does not
exist. Remark that, when they exist, these vertices are necessarily
unique because of the well-formedness assumptions. Finally, we say
that a vertex $v$ is \emph{unblocked} iff it has no outgoing edge, and
that it is \emph{final} iff $(i)$ $v$ is an \emph{unblocked task}
vertex and $(ii)$ $\state(v)=f_\lambda(v)$ (that is, $v$ represents a
task that has reached the final state of its transition system and is
not waiting on another task).

Let us now define several operations on \ctg. We will rely on these
operations when defining the formal semantics of \qdas. Let $\Aa$ be a
\qdas and $\Graph_\Aa=\tuple{V,E,\lambda,\allowbreak \queue,\state}$
be a \ctg for $\Aa$. Then:
\begin{myitemize}
\item for all $v\in V$: $\Graph\setminus v$ is the restriction of
  $\Graph$ to $V\setminus\{v\}$.
\item for all $\gamma\in\Gamma$ and $q\in\QID$,
  $\enqueue(q,\gamma)(\Graph_\Aa)$ is the \ctg $\tuple{V', E',
    \lambda', \queue', \state}$ where: $V'=V\cup\{v'\}$, $v'$ is a
  fresh queue vertex, $\lambda(v')=\gamma$, $\queue(v')=q$,
  $\state(v')=s^0_\gamma$, and for all $v\in V$:
  $\lambda'(v)=\lambda(v)$ and $\queue'(v)=\queue(v)$.  Finally,
  $E'=E\cup E_1\cup E_2$, where: $(i)$
  $E_1=\{(v',\tail(\Graph_\Aa,q))\}$ if
  $\tail(\Graph_\Aa,q)\neq\bot$,
  and $E_1=\emptyset$ otherwise, and $(ii)$ if $v\in V$ is a
  \emph{task} node s.t. $\queue(v)=q\in\SQID$, then $E_2=\{(v',v)\}$,
  otherwise $E_2=\emptyset$. Intuitively, this operation inserts a
  call to $\gamma$ in the queue $q$, by creating a new vertex $v'$ and
  adding an edge to maintain the FIFO ordering, if necessary (set
  $E_1$). In the case of a \emph{serial} queue that was empty before
  the enqueue, a supplementary edge (in set $E_2$) might be necessary
  to ensure that $v'$ is blocked by a currently running $v$ which has
  been extracted from $q$.
\item for all $q\in\QID$, if $\head(q)$ is different from $\bot$ and
  \emph{unblocked}, then $\dequeue(q)(\Graph_\Aa)$ is the \ctg
  $\tuple{V_C'\dotcup V'_T, E', \lambda, \queue, \state}$ where
  $V'_C=V_C\setminus\{\head(q)\}$ and
  $V'_T=V'_T\cup\{\head(q)\}$. Otherwise, $\head(q)=\bot$ and
  $\dequeue(q)(\Graph_\Aa)$ is undefined. Intuitively, this operation
  removes the first (with respect to the FIFO ordering) block from $q$
  and turns the corresponding \emph{call} vertex $\head(q)$ into a
  \emph{task} vertex, meaning that the block is now running as a
  task.\item for all $\delta=(s,a,s')\in\Delta$, $\step(\delta)(\Graph_\Aa)$
  is a \emph{set} of \ctg defined as follows. $\tuple{V, E, \lambda,
    \queue, \state'}\in\step(\delta)(\Graph_\Aa)$ iff there exists an
  \emph{unblocked} $v\in V_T$ s.t. $\state(v)=s$, $\state'(v)=s'$ and
  for all $v'\neq v$: $\state'(v')=\state(v')$. Remark that
  $\step(\delta)(\Graph_\Aa)$ can be empty. Intuitively, each graph in
  $\step(\delta)(\Graph_\Aa)$ corresponds to the firing of an
  $a$-labeled transition by a task that is not blocked.
\item for all unblocked $v\in V\cup\{\bot\}$, all $v'\in V$:
  $\letwait(v,v')(\Graph_\Aa)$ is either the \ctg $\Graph_\Aa$ if
  $v=\bot$, or the \ctg $\tuple{V, E\cup (v,v'),\lambda, \queue,
    \state}$ if $v\neq \bot$. Intuitively, this operation adds a wait
  edge between nodes $v$ and $v'$ when $v\neq\bot$, and does not
  modify the \ctg otherwise.
\end{myitemize}

\paragraph{\bf Semantics of \qdas:} For a \qdas $\Aa$ with set of variables
$\Xx$, a \emph{configuration} is a pair $(\Graph,\Data)$, where
$\Graph$ is a \ctg of $\Aa$ and $\Data\in\valof{\Xx}$.  The
operational semantics of $\Aa$ is given as a transition system
$\sem{\Aa}$ whose states are configurations of $\Aa$; and whose
transitions reflect the semantics of the actions labeling the
transitions of the \qdas. Formally, given a \qdas $\Aa=\tuple{\CQID,
  \SQID, \Gamma, main, \Xx, \Sigma, (\Ts_\gamma)_{\gamma\in\Gamma}}$,
$\sem{\Aa}$ is the labeled transition system $\langle C, c^0, \SSigma,
\Longrightarrow \rangle$ where: $(i)$ $C$ contains all the pairs
$(\Graph,\Data)$ where $\Data\in\valof{\Xx}$, and $\Graph$ is a \ctg
of $\Aa$, $(ii)$ $c^0=(\Graph^0,\Data^0)$ with $\Data^0(x)=d_0$ for
all $x\in\Xx$, and $\Graph^0=\tuple{\{v^0\},\emptyset,\lambda,
  \queue,\state}$, where $v^0$ is a \emph{task node},
$\lambda(v^0)=main$, $\state(v^0)=s^0_{\text{main}}$ and
$\queue(v^0)=\imath$, $(iii)$ $\SSigma=\Sigma\dotcup\{\e\}$ and $(iv)$
$\big((\Graph,\Data),a,(\Graph',\Data')\big)\in\Longrightarrow$ iff
one of the following holds:
\begin{description}
\item[Async. dispatch:] $a=\dispa(q,\gamma)$, $\Data'=\Data$, and
  there are $\delta=(s,a,s')\in \Delta$ and
  $\Graph''\in\step(\delta)(\Graph)$ s.t.:
  $\Graph'=\enqueue(q,\gamma)(\Graph'')$.
\item[Sync. dispatch:] $a=\disps(q,\gamma)$, $\Data'=\Data$ and there
  are $\delta=(s,a,s')\in \Delta$ and
  $\Graph''\in\step(\delta(\Graph))$
  s.t.: $\Graph'=\letwait(v,v')\big(\enqueue(q,\gamma)(\Graph'')\big)$
  where $v$ is the node whose $\state$ has changed during the $\step$
  operation, and $v'$ is the fresh node that has been created by the
  $\enqueue$ operation. That is, a queue vertex $v'$ labeled by
  $\gamma$ is added to $q$ and a \emph{wait} edge is added between the
  node $v$ representing the task that performs the \emph{synchronous}
  dispatch, and $v'$, as the dispatch is \emph{synchronous}.
\item[Test:] $a=g\in\guardson{\Xx}$, $\Data'=\Data$, $\Data\models g$,
  and there is $\delta=(s,a,s')\in \Delta$ s.t.
  $\Graph'\in\step(\delta)(\Graph)$.
\item[Assignment:] $a=x\gets v\in\assignon{\Xx}$, $\Data'(x)=v$, for
  all $x'\neq x$: $\Data'(x)=\Data(x)$ and there is
  $\delta=(s,a,s')\in \Delta$
  s.t. $\Graph'\in\step(\delta)(\Graph)$.
\item[Scheduler action:] $a=\e$, $\Data'=\Data$ and:
  \begin{itemize}
  \item either there is a final vertex $v$
    s.t. $\Graph'=\Graph\setminus v$;
  \item or there is $q\in\CQID$ s.t. $\head(q,\Graph)\neq\bot$ and
    $\Graph'=\dequeue(q)(\Graph)$. That is, the scheduler schedules a
    block (represented by $v$) from a concurrent queue.
  \item or there is $q\in \SQID$ s.t.  $\head(q,\Graph)=v$, $v$ is
    \emph{unblocked}, as well as $\Graph'=\letwait(\head(q,G''),v)(G'')$ and
    $G''=\dequeue(q)(\Graph)$. That is, the scheduler schedules a
    block (represented by $v$) from the serial queue $q$. As the queue
    is serial, a \emph{wait} edge is inserted between the next waiting
    block in $q$ (now represented by $\head(q,G'')$) and $v$.
  \end{itemize}

 \end{description}

 A \emph{run} $\rho$ of a \qdas is an alternating sequence $c_0
 a_1 c_1 a_2\dots a_n c_n$ of configurations and actions where
 $(c_i,a_{i+1},c_{i+1})\in\Longrightarrow$ for all $0\leq i < n$ and
 $c_0=c^0$. A run is \emph{finite} if this sequence is
 finite.
 A configuration $c$ is \emph{reachable} in $\Aa$ iff there
 exists a finite run $c_0 a_1 c_1 a_2\dots a_n c_n$ of $\Aa$ s.t. $c_n=c$. We
 denote by $\Reach(\Aa)$ the set of all reachable configurations
 of~$\Aa$.

 The decision problem on \qdas we mainly consider in this work is the
 \emph{Parikh coverability problem}: given a \qdas $\Aa$ with set of
 locations $S$ and a function $f:S\mapsto \NN$, it asks whether
 there is $c=(\Graph,\Data)\in\Reach(\Aa)$
 s.t. $f\mleq\Parikh(\Graph)$. When the answer to this question is `yes',
 we say that $f$ is \emph{Parikh-coverable} in $\Aa$. It is well-known
 that meaningful verification questions can be reduced to this
 problem. For instance, consider a \emph{mutual exclusion} question,
 asking whether it is possible to reach, in a \qdas $\Aa$, a
 configuration in which at least two tasks are executing the same
 block $\gamma$ and are in the same control state $s$. If yes, the
 mutual exclusion (of control state $s$) is violated. This can be
 encoded into an instance of the Parikh coverability problem, where
 $f(s)=2$ and $f(s')=0$ for all $s'\neq s$, and would allow,
 for example, to verify if there are more than one block of type
 \texttt{increase} running in
 Example\,\ref{ex:gcd-matrixmult}.

 In addition, we look at the \emph{(universal) termination problem}:
 given a \qdas $\Aa$, it asks whether all executions of $\Aa$ are
 finite, i.e., there is no infinite run of $\Aa$.
 Regarding Example\,\ref{ex:gcd-matrixmult}, this
 permits to test whether the \texttt{main} task terminates, i.e.,
 all dispatched blocks terminate.

\section{From the Parikh coverability problem to Termination\label{sec:parikh-cover-probl}}

Before regarding the termination problem, we first study
in this section the Parikh coverability problem from a
computational point of view. As expected, this problem is undecidable
in general. However, when restricting the types of queues and
dispatches that are allowed, it is possible to retain decidability. In
these cases, we characterize the complexity of the problem. Formally,
we consider the following subclasses of \qdas. A \qdas $\Aa$ with set
of transitions $\Delta$, set of serial queues $\SQID$ and set of
concurrent queues $\CQID$, is \emph{synchronous} iff there exists no
$(s,a,s')\in\Delta$ with $a\in\{\cfont{dispatch\_a}\}\times \QID
\times \Gamma$; it is \emph{asynchronous} iff there exists no
$(s,a,s')\in\Delta$ with $a\in\{\cfont{dispatch\_s}\}\times \QID
\times \Gamma$; it is \emph{concurrent} iff $\SQID=\emptyset$ and
$\CQID\neq\emptyset$; it is \emph{serial} iff $\CQID=\emptyset$ and
$\SQID\neq\emptyset$; it is \emph{queueless} iff
$\CQID=\SQID=\emptyset$.

\paragraph{\bf Queueless \qdas:} In a queueless \qdas, there is no dispatch
possible, so the only task that can execute at all time is the
\texttt{main} one. Thus, configurations of queueless \qdas can be
encoded as tuples $(s,\Data)$, where $s$ is a state of
\texttt{main}, and $\Data$ is a valuation of the variables.
Hence queueless \qdas are essentially \lts with variables over
a finite data domain, thus:
\begin{proposition}\label{prop:queueless}
  The Parikh coverability is \pspacecomplete
  for queueless \qdas.
\end{proposition}

\paragraph{\bf Synchronous \qdas:}

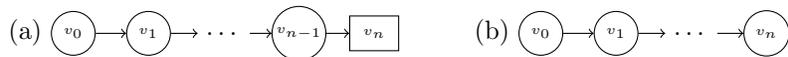
\begin{figure}[!t]
  \centering(a)
  \begin{tikzpicture}[baseline=-0.5ex,every node/.style={inner sep=3pt}]
    \foreach \x / \z in {0/0,1/1}
      \draw(1*\x,0) node[circle,draw] (\x) {\tiny{$v_\z$}};
    \draw (2,0) node (2) {\dots};
    \draw(3,0) node[circle,draw,inner sep=1pt] (3) {\tiny{$v_{n-1}$}};
    \draw(4,0) node[rectangle,inner sep=5pt,draw] (4) {\tiny{$v_n$}};
    \foreach \x / \y in {0/1,1/2,2/3,3/4}
    \draw (\x) edge [->] (\y);
  \end{tikzpicture}\hspace*{1cm}(b)
\begin{tikzpicture}[baseline=-0.5ex,every node/.style={inner sep=3pt}]
    \foreach \x / \z in {0/0,1/1}
      \draw(1*\x,0) node[circle,draw] (\x) {\tiny{$v_\z$}};
    \draw(3,0) node[circle,draw] (3) {\tiny{$v_n$}};
    \draw (1*2,0) node (2) {\dots};
    \foreach \x / \y in {0/1,1/2,2/3}
    \draw (\x) edge [->] (\y);
  \end{tikzpicture}
  \caption{The two possible forms of reachable {\ctg}s in a synchronous \qdas}
  \label{fig:config-sync}
\end{figure}

In synchronous \qdas, there is no concurrency in the sense there is at
most one running task that can fire an action at all times. All the
other tasks have necessarily performed a \emph{synchronous} dispatch
and are thus blocked. More precisely, in every reachable configuration
$(\Graph,\Data)$ of a synchronous \qdas, $\Graph$ is of one of the
forms depicted in Fig.~\ref{fig:config-sync}
(i.e. $v_0,\dots,v_{n-1}\in V_T$ and either $v_n\in V_T$ or $v_n\in
V_C$). When the current \ctg is of the form
Fig.~\ref{fig:config-sync}(a), the only possible action is that the
scheduler starts running $v_n$'s block and we obtain a graph of the
form Fig.~\ref{fig:config-sync}(b). In the case where the \ctg is of
the form (a), either $v_n$ terminates, which removes $v_n$ from the
\ctg, or $v_n$ executes an internal action, which does not change the
shape of the \ctg, or $v_n$ does a synchronous call, which adds a call
vertex as successor of $v_n$ which will be directly
scheduled. W.l.o.g., we assume in the following that for synchronous
\qdas the combined action of \disps and scheduling the dispatched
block is atomic.

                                For a \ctg $\Graph$ and $w\in S^*$, we write
 $\Graph\triangleright w$ iff for all $0\leq i\leq n$:
$w_i=\state(v_i)$ and the empty \ctg is mapped to the empty word $\e$.
Given a synchronous \qdas $\Aa$ with set of local
states $S$ as before, we can build a pushdown system with data
$\Pp_\Aa$ such that, at all times, the current location of $P_\Aa$ encodes the
current location of the (single) running block in $\Aa$, and the stack
content records the sequence of synchronous dispatches, as described
above. A guard or assignment in $\Aa$ is kept as is in $\Pp_\Aa$. A
synchronous dispatch $(s,\disps(q,\gamma),s')$ in $\Aa$ is simulated
by a push of $s'$ (to record the local state that has to be reached
when the callee terminates) and moves the current state of $\Pp_\Aa$ to
the initial state of $\gamma$. The termination of a block is simulated
by a pop (and we encode the termination of
$main$ in testing the stack's emptiness).

\begin{restatable}{proposition}{propsyncqdaspds}
\label{prop:sync_qdas_pds}  Given a synchronous \qdas $\Aa$, then we can construct  a
  pushdown system with data $\Pp_\Aa$ such that the following holds:
  for any
  run $\rho=c_0a_1c_1\dots a_nc_n$ of $\Aa$,
  there exists a run $\pi=x_0a_1x_1\dots a_nx_n$ in $\Pp_\Aa$ such that
  for all $c_i=(\Graph_i,\Data_i)$ and $x_i=(s_i,w_i,\Data'_i)$
  we have $\Data_i=\Data_i'$ and $\Graph_i\triangleright w_i$ ($0\leq i \leq n$),
  and vice versa.
\end{restatable}
\inputproof{proofs/lem:sync_qdas_pds.tex}

The previous proposition allows to derive results on the reachability
problem. However, we are interested in the Parikh coverability
problem. Let $f$ be a Parikh image of $\Aa$. Then, by
Proposition~\ref{prop:syncqdas_parikh_sim}, looking for a reachable
configuration of $\Aa$ that covers $f$ amounts to finding a reachable
configuration $(s_i,w_i,\Data_i)$ of $\Pp_\Aa$ s.t. the Parikh image
$P$ of $w_i$ is s.t.  $f\mleq P$ (as the \ctg is encoded by the stack
content $w_i$). To achieve this, we augment $\Pp_\Aa$ with a
\emph{widget} that works as follows. In any location of $\Pp_\Aa$, we
can jump non-deterministically to the widget. Then, the widget pops
all the values from the stack, and checks that at least $f(s)$ symbols
$s$ are present on the stack. The widget jumps to an accepting state
iff it is the case.  We call $\Pp_{\Aa,f}$ the resulting
\pds. Clearly, one can build such a widget for all $f$, and this
effectively reduces the Parikh coverability problem of \qdas to the
location reachability problem of \pds. Moreover, for all $f$, the
widget is of size exponential in $\vert S\vert$ and exponential
in the binary encoding of
 $max_{s\in S}f(s)$. Hence, building $\Pp_{\Aa,f}$
requires exponential time:

\begin{restatable}{proposition}{propsyncqdasparikhsim}
  \label{prop:syncqdas_parikh_sim}
  Given a synchronous \qdas $\Aa$ with states $S$ and a function
  \mbox{$f:S\rightarrow \NN$}, then one can generate a \pds
  $\Pp_{\Aa,f}$ of size exponential in $\Aa$ and a state $s$ of
  $\Pp_{\Aa,f}$, s.t. $\Pp_{\Aa,f}$ reaches $s$ iff $f$ is Parikh
  coverable in $\Aa$.
\end{restatable}

As testing emptiness of a pushdown system without data is
\ptimecomplete~\cite{bouajjani-a-1997-135-a}, the
Parikh coverability problem is in \dexptime for \emph{synchronous}
\qdas (with both types of queues). A matching lower bound is obtained
by reducing the reachability question of \pds with data (see
Proposition~\ref{prop:reach_pds_data}). This reduction requires only
one \emph{concurrent} queue, so the Parikh reachability problem is
\dexptime-hard for \emph{synchronous concurrent} \qdas. Hence we
derive the following:

\begin{theorem}\label{thm:syncqdas}
  The Parikh coverability problem is \dexptimecomplete for synchronous
  and for synchronous concurrent \qdas.
\end{theorem}

Let us take a closer look on the dispatches that happen in runs of
synchronous \qdas that have only \emph{serial} queues.
Here, each task except the \texttt{main} task blocks the queue it is
started
from. Hence, any other block dispatched to these already blocked
queues deadlocks. Thus, all reachable \ctg have at most
$\vert SQID\vert +2$ vertices. Hence, the pushdown systems used in
all previous constructions have bounded stack height, and we can
apply test on a finite transition system.
 The lower bound can be
derived from Proposition\,\ref{prop:queueless}.
by testing the emptiness of the intersection
of $n$ finite processes, that is $\pspace$-complete~\cite{kozen-d-1977-254-a}.

\begin{theorem}\label{thm:serialsyncqdas}
  The Parikh coverability problem is \pspacecomplete for serial
  synchronous \qdas.
\end{theorem}

\paragraph{\bf Concurrent asynchronous \qdas:}

Let us now establish a relationship between \emph{concurrent
asynchronous {\qdas}} and Petri nets that proves that the Parikh
coverability problem is \dexpspace-complete. We first show how to
reduce the \qdas Parikh coverability problem to the Petri net
coverability problem.
                    Given a concurrent asynchronous \qdas $\Aa$, we construct a Petri net $N_\Aa$ as
follows: The places of $N_\Aa$ are $(\Xx\times\DD)\cup S$.
Each place $s\in S$  counts how many blocks are currently running and are in state $s$. Each
place $(x,d)$ encodes the fact that variable $x$ contains value $d$ in
the current valuation.  Remark that we have no place to encode the
contents of the queue, as the dispatch of block $\gamma$ directly
creates a new token in $s^0_\gamma$. This encoding is, however,
correct with respect to to the \emph{Parikh coverability problem}, as $\Parikh(G)$
does not distinguish between a block $\gamma$ that is waiting in a
queue, and a task executing $\gamma$ in its initial state. Thus:

\begin{restatable}{proposition}{propfromqdastopn}
  \label{prop:from-qdas-to-pn}
  For all concurrent asynchronous \qdas $\Aa$ with set of location
  $S$, we can build, in polynomial time, a Petri net $N_\Aa$ s.t. $f$ is
  Parikh-coverable in $\Aa$ iff $m\in\Cover(N_\Aa)$, where $m$ is the
  marking s.t. for all $s\in S$: $m(s)=f(s)$ and for all $p\in
  P\setminus S$: $m(p)=0$.
\end{restatable}

Let us now reduce the Petri net coverability problem to the \qdas
Parikh coverability problem. Let $N=\tuple{P,T,m_0}$ be a Petri net. We
associate to $N$ the concurrent asynchronous \qdas
${\Aa}_N=\tuple{\CQID, \emptyset, \Gamma, \mathtt{main}, \Xx, \Sigma,
  (\Ts_\gamma)_{\gamma\in\Gamma}}$, on the finite domain
$\DD=\{0,1\}$, where $\CQID=\{C\}$, $\Gamma=\{main,trans\}\cup P$,
$\Xx=\{v_p\mid p\in P\}$ and $(\Ts_\gamma)_{\gamma\in\Gamma}$ is given
by the pseudo-code in Fig.~\ref{fig:simulossyPN} (this construction is
an extension of a construction found in \cite{ganty-p-2010--a}).  We
assume that, for $\gamma\in\{\mathtt{trans},\mathtt{main}\}$
$\mystate{\ell}{\gamma}$ is the location of $\gamma$'s \lts that is
reached when the control reaches line $\ell$. Let $G=\tuple{V, E,
  \lambda, \queue, \state}$ be a \ctg for $\Aa_N$, and let $m$ be a
marking of $N$. Then, we say that \emph{$G$ encodes $m$}, written
$G\rhd m$ iff $(i)$
$\Parikh(G)(\mystate{14}{trans})=\Parikh(G)(\mystate{8}{main})=1$,
$(ii)$ for all $p\in P$: $\Parikh(G)(s^0_p)=m(p)$ and $(iii)$ for all
$p\in P$, for all $s\in S_p\setminus\{s^0_p\}$:
$\Parikh(G)(s)=0$. Thus, intuitively, a \ctg $G$ encodes a marking $m$
iff \texttt{main} is at line 8, \texttt{trans} is at line 14, $m(p)$
counts the number of \texttt{p} blocks that are either in $C$ or
executing but at their initial state, and there are no \texttt{p}
blocks that are in state $s_p^{mid}$ or $s_p^{fin}$.

\begin{figure}[!t]
  \centering
  \begin{minipage}[t]{.45\linewidth}
    \begin{lstlisting}[numberblanklines=false]
def main():
  for each $p\in P$:
    $v_p$ := 0
    select $k_p\in\{0,\ldots,m_0(p)\}$
    for i = 0...$k_p$:
      dispatch_a($C$, p())
  dispatch_a($C$, trans())
  while(true): do nothing (*@\mbox{ }\\[-0.5ex]@*)
block p(): // For all $p\in P$
  while($v_p = 0$): do nothing
  $v_p$ := 0
\end{lstlisting}
\end{minipage}
\begin{minipage}[t]{.45\linewidth}
  \begin{lstlisting}[firstnumber=last]
block trans():
  while(true):
    select $t=(I_t,O_t)\in T$
    for each $p\in P$ s.t. $I_t(p)=1$:
      $v_p$ := true
    while($\exists p\in P$: $v_p=1$): do nothing
    for each $p\in P$ s.t. $O_t(p)=1$:
      dispatch_a($C$, p())
    \end{lstlisting}
  \end{minipage}
\caption{Encoding of Petri net coverability $\tuple{P, T, m_0}$ by a \qdas}
  \label{fig:simulossyPN}
\end{figure}

The intuition behind the construction is as follows. Each run of the
\qdas $\Aa_N$ starts with an initialization phase, where \texttt{main}
initializes all the $v_p$ variables to $0$ and dispatches, for all
$p\in P$, $k_p$ blocks \texttt{p} with $k_p\leq m_0(p)$, then
dispatches a call to \texttt{trans}. At that point, the only possible
action is that the scheduler dequeues all the blocks. All the
\texttt{p} tasks are then blocked, as they need that $v_p=1$ to
proceed and terminate. Then, \texttt{trans} cyclically picks a
transition $t$, sets to $1$ all the variables $v_p$ s.t.  $t$ consumes
a token in $p$, and waits that all the $v_p$ variables return to
$0$. This can only happen because \emph{at least} $I_t(p)$ \texttt{p}
tasks have terminated, for all $p\in P$. So, when \texttt{trans}
reaches line 19, the encoded marking has been decreased by \emph{at
  least} $I_t$. Remark that more than $I_t(p)$ \texttt{p} tasks could
terminate, as they run concurrently, and the lines 11 and 12 do not
execute atomically. Then, \texttt{trans} dispatches one new \texttt{p}
block iff $t$ produces a token in $p$. This increases the encoded
marking by $O_t$, so the effect of one iteration of the main
\texttt{while} loop of \texttt{trans} is to simulate the effect of
$t$, plus a possible token loss. Hence, the resulting marking is
guaranteed to be in $\Cover(N)$ (but maybe not in $\Reach(N)$). This
is formalized by the following proposition:

\begin{restatable}{proposition}{propfrompntoqdas}\label{prop:from-pn-to-qdas}
  For all Petri nets  $N$, we can build, in polynomial time, a concurrent
  asynchronous \qdas $\Aa_N$ s.t.  $m\in\Reachloss(N)$ iff there
  exists $(G,\Data)\in\Reach(\Aa_N)$ with $G\rhd m$.
\end{restatable}

\begin{theorem}\label{thm:concasyncqdas}
  The Parikh coverability problem
  is \dexpspace-complete for concurrent asynchronous \qdas.
\end{theorem}

\paragraph{\bf Asynchronous Serial \qdas:}

Let us show that for the class
of \qdas with one serial queue, and where asynchronous dispatches are
allowed, the Parikh coverability problem is \emph{undecidable}. We
establish this by a reduction from the control-state  reachability
problem in a \fifo system which is known to be
undecidable~\cite{brand-d-1983-323-a}.

Intuitively, we use the serial queue to
model the unbounded, reliable fifo queue where sending a message $m$
is encoded as asynchronously dispatching a block $\gamma_m$. This
block $\gamma_m$ contains the control-flow of
receiving $m$, i.e., that will resume the \fifo system's execution
directly after receiving $m$. The \fifo system's global state is
guarded in a global variable. Receiving a certain message $m$
is encoded as terminating the currently running task and assuring
(via a global variable) that the succeeding task's type is the one of
the expected message.

\begin{theorem}\label{the:async-seri-undec}
  The Parikh coverability problem is undecidable for asynchronous
  \qdas with at least one serial queue.
\end{theorem}

\paragraph{\bf Concurrent \qdas:}
Let us show that, once we allow both
synchronous and asynchronous dispatches in a \emph{concurrent} \qdas,
the Parikh coverability problem becomes undecidable.
For that purpose, we reduce the reachability problem of two counter
systems.

The crux of the construction is the use of
variables, i.e., global memory, to implement a rendez-vous
synchronization. Given two distinct tasks, one can use their nested
access to two lock variables to guard a shared data variable by
assuring that a value written to the variable must be read before it
is overwritten.

Let us give the construction's intuition: Each counter is
encoded similarly to the construction for synchronous \qdas as
pushdown stack over a singleton alphabet, i.e., a sequence of nested
synchronous dispatched blocks, these are controlled via rendez-vous
 from the \main task that in the beginning
asynchronously dispatched the two counters.

\begin{theorem}\label{thm:concqdasundec}
  The Parikh coverability problem is undecidable for concurrent \qdas that
  use both synchronous and asynchronous dispatches.
\end{theorem}

\paragraph{\bf Termination Problem:}

We use the previous constructions to directly
lift the undecidability results from the Parikh coverability problem to the termination problem.
The close connection of synchronous \qdas with \pds (with data)
allows to directly derive an \dexptime algorithm for the termination problem from the emptiness testing of
Büchi \pds~\cite{esparza-j-2000-232-a}. Up to our knowledge,
no completeness result is known for the latter problem, thus leaving a gap to the directly derivable \pspace-hardness via finite systems.
The result for  asynchronous concurrent
\qdas directly follows from Petri nets~\cite{lipton,rackoff}.

\begin{theorem}\label{thm:termination}
  The termination problem is  \pspacecomplete for synchronous serial
  \qdas, it is in \dexptime and \pspace-hard for synchronous \qdas,
  and it is \dexpspacecomplete for asynchronous concurrent \qdas.
  It is undecidable for asynchronous serial \qdas, and \qdas that
  use both synchronous and asynchronous dispatches.
\end{theorem}

\section{Extending QDAS with Fork/Join\label{sec:forkjoin}}

We return to the introductory matrix multiplication example. The crux of the
algorithm is the parallel for-loop that \emph{forks} a finite number of subtasks and waits
for their termination (\emph{join}).  The latter had to be implemented via a global
semaphore which $(i)$~restricts the number of forkable tasks by the underlying
finite value domain, and $(ii)$~needs to be properly guarded by the programmer for access outside fork and join.
In the following we thus
want to extend \qdas by an explicit fork/join construct (which also exists in GCD).
Further, the given matrix multiplication algorithm depended on an a priori fix size
for the factor matrices, however,
in practice, one wants to verify the algorithm for any possible (correct) input
of any size. Thus, we need to consider the verification of extended
\qdas  where the number of forked tasks is parametrized by the
input.

As fork/join behaviour relies on asynchronously dispatching tasks
on a concurrent queue, we ignore in the following
synchronous dispatches and serial queues, thus also partially
avoiding the
previous basic undecidability results. Note that asynchronous
concurrent \qdas can be regarded as over-approximations of all other
classes of \qdas.

\paragraph{\bf QDAS extended by fork/join}
An \emph{\qdas extended by fork/join (\eqdas)}
is a tuple $\langle \CQID, \emptyset,
  \Gamma, main, \Xx, \Sigma, (\Ts_\gamma)_{\gamma\in\Gamma}\rangle$
  that is equivalent to a \qdas except that we replace in $\Sigma$
  the synchronous dispatch by the following action:
  $\{\forkjoin\}\times \CQID\times \Gamma \times (\NN\cup\{\ast\})$.
  The \emph{parameter} of a \forkjoin action is the last value of the
  tuple.
  An \eqdas is \emph{$\ast$-free} if in all $\Ts_\gamma$ for
  $\gamma\in\Gamma$ the parameter of the \forkjoin action is
  not~$\ast$.

  The semantics of an \eqdas
  is given analogous to standard \qdas as
  transition system $\langle C, c^0,$ $ \Sigma, \Longrightarrow \rangle$
  where we additionally extend the transition relation
  $\Longrightarrow$ given by tuples
$\big((\Graph,\Data),a,(\Graph',\Data')$
by the following case:
\begin{description}
  \item[Fork/join:] $a=\forkjoin(q,\gamma,p)$ with
    $p\in(\NN\cup\{\ast\})$,
    $\Data'=\Data$ and there
  are $\delta=(s,a,s')\in \Delta$,  and
  $\Graph''\in\step(\delta(\Graph))$
  such that:
  if $p=\ast$ then we choose non-deterministically an $n\in\NN$, else
  $n=p$, so that
  $\Graph'=\Graph''_n$ where
  $\Graph''_0=\Graph''$ and for $0<i\leq n$ we define
  $\Graph''_{i+1}=\letwait(v,v_{i+1}')\big(\enqueue(q,\gamma_{i+1})(\Graph''_i)\big)$
  where $v$ is the node whose $\state$ has changed during the $\step$
  operation, and $v'_{i+1}$ is the fresh node that has been created by the
  $\enqueue$ operation.
\end{description}
Intuitively, a \forkjoin action appends a sequence of blocks to a
queue by additionally adding a wait edge to each newly create node.
Hence, the join is modeled by a separate action that is taken by the
scheduler after deleting the wait edges.

The \emph{extended Parikh
coverability problem} asks, given an \eqdas $\Aa$ with locations
$\Ss$ and a mapping $f:\Ss\rightarrow\NN$, whether
there exists $c=(\Graph,\Data)\in\Reach(\Aa)$ with
$f\mleq\Parikh(\Graph)$. The \emph{extended termination problem}
asks, given an \eqdas $\Aa$ whether
there is no infinite run possible in $\Aa$.

As \forkjoin actions with parameter $1$ are semantically
equivalent to a synchronous dispatch action, we can directly
reduce the two counter machine simulation from the proof of
Theorem\,\ref{thm:concqdasundec}  to \eqdas.

\begin{theorem}\label{thm:forkjoin_undec}
  Both  the extended Parikh coverability and extended termination problem are undecidable.
\end{theorem}

Consequently, we focus on two distinct over-approximations for \eqdas in
the following that allow us to give  approximative answers to our
verification problems.

\paragraph{\bf $\ast$-free \eqdas:}
Given an \eqdas $\Aa$ that is $\ast$-free. We construct
a Petri net $N_\Aa^\times$ by extending the previous
construction from asynchronous concurrent \qdas to Petri nets as
follows: As in the \eqdas semantics we split a single \forkjoin
action of a block $\gamma$ on a queue $q$ with parameter $n\in\NN$
into $(i)$~a fork transition that creates $n$ new tokens in
$s_\gamma^0$, and $(ii)$~a subsequent join transition that depends on
taking $n$ tokens from the place representing $f_\gamma$.
Analogous to the proof of Proposition\,\ref{prop:from-qdas-to-pn} we can show
the following:

\begin{restatable}{proposition}{propastfreeapprox}
  \label{prop:astfreeapprox}
  For all $\ast$-free \eqdas with set of location $\Ss$, we can build
  in polynomial time a Petri net $N_\Aa^\times$ st. $f$ is
  Parikh-coverable in $\Aa$ if $m\in\Cover(N_\Aa^\times)$, where $m$
  is the marking s.t. for all $s\in S$: $m(s)=f(s)$ and for all $p\in
  P\setminus S$: $m(p)=0$. Further, if
  $N_\Aa^\times$ terminates, then $\Aa$ is guaranteed to terminate.
\end{restatable}

As coverability and termination are decidable for Petri nets, we can
decide extended Parikh coverability and extended termination
on this over-abstraction.

\paragraph{\bf \eqdas with $\ast$ parametrized fork/join:}
Given an \eqdas $\Aa$ that is not $\ast$-free, we construct
a Petri net $N_\Aa^\ast$ as follows
starting from the construction for
asynchronous concurrent \qdas: For \forkjoin actions
whose parameter is not $\ast$, we proceed as in the above
construction for $\ast$-free \eqdas. However, we need to model
the forking of an arbitrary number of blocks when the parameter of
the \forkjoin action equals $\ast$. For this, we use Petri nets
extended with $\omega$-arcs. An outgoing arc of a transition labeled
with $\omega$ adds an arbitrary number of tokens to the corresponding
place, thus, we translate the fork of block $\gamma$ into an
$\omega$-transition leading to place $s_\gamma^0$. The join is
approximated by a transition that non-deterministically chose to
advance the original workflow, ignoring not already terminated forked
tasks. Thus by extending the proof of
Proposition\,\ref{prop:from-qdas-to-pn}:

\begin{restatable}{proposition}{propwithastapprox}
  \label{prop:withastapprox}
  For all \eqdas with set of location $\Ss$, we can build
  in polynomial time a Petri net $N_\Aa^\ast$ st. $f$ is
  Parikh-coverable in $\Aa$ if $m\in\Cover(N_\Aa^\ast)$, where $m$ is the
  marking s.t. for all $s\in S$: $m(s)=f(s)$ and for all $p\in
  P\setminus S$: $m(p)=0$. Further, if $N_\Aa^\ast$ terminates,
  then $\Aa$ is guaranteed to terminate.
\end{restatable}

We have recently shown that the termination problem is decidable for
Petri nets with $\omega$-arcs~\cite{geeraerts-g-2012--a}.
Hence, also  extended termination is decidable on the previous
abstraction.

With respect to coverability, we can replace the $\omega$-arcs of
$N_\Aa^\ast$ by a non-deterministic loop that adds an arbitrary
number of tokens to the original arc's target place. Note that this
simple trick does not work for verifying termination. Consequently,
we can use the known algorithms for coverability on this polynomially
larger standard Petri net, and hence  the extended Parikh
coverability problem is decidable on this abstraction.

\section{Conclusion \& Outlook}
We introduce the, up to our knowledge, first formal model that grasps
the core of \gcd, and that allows to derive basic results on the
decidability of verification question thereupon. Due to the obvious
undecidability issues of the model, we currently focus on several
under- and over-approximative approaches (e.g., language bounded
verification, graph minor based abstractions, novel Petri net
extensions~\cite{geeraerts-g-2012--a}) as well as enhancements for
additional \gcd features like task groups, priorities, and timer
events.

\newcommand{\nb}{}

\clearpage
\appendix

\section{Proof for Section~\ref{sec:prelims}}
\propreachpdsdata*
\begin{proof}
  For the upper bound, we generate a reachability-equivalent
  \pds (without data)
  by encoding all possible data valuations into the pushdown system's
  states. This leads to an exponential blowup of the state space.
  The lower bound can be derived from the reduction
of the emptiness test of the intersection of a context-free language
with
$n$ regular languages that is known to be \dexptime-hard (hardness follows
easily by a reduction from linearly bounded alternating Turing machines;
a closely related problem, the reachability of pushdown systems with
checkpoints, is shown to be \dexptime-hard
in (*).\end{proof}

{\small
(*)~Javier Esparza, Anton\'{\i}n Ku\v{c}era, and Stefan Schwoon:
\textsl{Model checking  LTL with regular valuations for pushdown systems},
in { Information and Computation}, 186(2):355--376, 2003.
}

\section{Proofs of Section\,\ref{sec:parikh-cover-probl}}

\paragraph{\bf Synchronous \qdas:}

Let $\Aa$ be a synchronous \qdas with a set of locations $S$,
a set of rules $\Delta$, a set of final states $F$, and
set of queues $\SQID$. Let $\Graph$ be a \ctg of one of the forms given
in Fig.~\ref{fig:config-sync}, and let $w=w_0w_1\cdots w_n$ be a word
in $S^*$. Then, \emph{$\Graph$ is encoded by $w$},
written $\Graph\triangleright w$, iff for all $0\leq i\leq n$:
$w_i=\state(v_i)$ and the empty \ctg is mapped to the empty word $\e$.

Given a synchronous \qdas $\Aa=\tuple{ \CQID, \SQID, \Gamma, main,
  \Xx, \Sigma, (\Ts_\gamma)_{\gamma\in\Gamma}}$ with set of local
states $S$ as before, we build a pushdown system with data
$\Pp_\Aa=\tuple{Y,\Xx,y^0,S,\Sigma_\Pp,\Delta_\Pp}$ where:
\begin{myitemize}
\item the set of states is $Y=S\cup\{\e\}$ and the initial state is
  $y^0=s^0_{\main}$
\item $\Sigma_\Pp =\left(\{\push,\pop\}\times S \right) \cup
  \{\emptystack\} \cup \guardson{\Xx}\cup\assignon{\Xx}$
\item a tuple $(y,a,y')$ is a transition rule in $\Delta_\Pp\subseteq
  Y\times \Sigma_\Pp \times Y$ iff
  \begin{itemize}
    \renewcommand{\labelitemi}{\bf$\sim$}
  \item $a\in\guardson{\Xx}\cup\assignon{\Xx}$ and
    $(y,a,y')\in\Delta$
  \item $a=\push(s')$, $(s,\disps(q,\gamma),s')\in\Delta$ and
    $y'=s^0_\gamma$
  \item $a=\pop(s)$, $y\in F$ and $y'=s$
  \item $a=\emptystack$, $y=f_{main}$, and $y'=\e$.
  \end{itemize}
\end{myitemize}
Thus, at all times, the current location of $P_\Aa$ encodes the
current location of the (single) running block in $\Aa$, and the stack
content records the sequence of synchronous dispatches, as described
above. A guard or assignment in $\Aa$ is kept as is in $P_\Aa$. A
synchronous dispatch $(s,\disps(q,\gamma),s')$ in $\Aa$ is simulated
by a push of $s'$ (to record the local state that has to be reached
when the callee terminates) and moves the current state of $P_\Aa$ to
the initial state of $\gamma$. The termination of a block is simulated
by a pop (and we use the $\emptystack$ action for the termination of
$main$).

\propsyncqdaspds*
\begin{proof}
  We assert that the semantics of $\Pp_\Aa$ is the usual semantics for
  pushdown systems with data, i.e., an infinite transition system with
  configurations $c=(y,w,d)\in Y\times S^* \times \DD^\Xx$.
  Thus, we can interpret configurations also as follows:
  $(x,d)\in S^*\times \DD^\Xx$ with $x=w\cdot y\in S^*\cdot (S\cup\{\e\}$.

  Let $(\Graph,\Data)\in\Reach(\Aa)$ be reachable by a run
  $(\Graph_0,\Data_0)a_1(\Graph_1,\Data_1)a_2\dots a_n(\Graph_n,\Data_n)$.
  Then we can induce a run $(x_0,\DData_0)a_1(x_1,\DData_1)a_2\dots
  a_n(x_n,\DData_n)$ in $\Pp_\Aa$ such that $\Data_i=\DData_i$ and
  $\Graph_i\triangleright x_i$ for $0\leq i \leq n$.

  By construction of $\Pp_\Aa$, $x_0\triangleright G_0$ and $\Data_0=\DData_0$.
  We now assume that there exists a prefix of the \qdas's run of length $0\leq
  j\leq n$ of the form
  $(\Graph_0,\Data_0)\dots
  (\Graph_j,\Data_j)$ such that there exists a run of the pushdown system
  $(x_0,\DData_0)\dots(x_j,\DData_j)$ that fullfills the induction hypothesis.
  We now consider the outcome of a \qdas transition labeled $a_{j+1}$.
  We know that $\Graph_j$ must be a path of vertices $v_0\dots v_n$ connected
  by wait edges.

\begin{description}
\item[Sync. dispatch:] dispatching a block $\gamma$ on queue $q$
  leads to $(\Graph_{j+1},\Data_{j+1})$ with $\Data_{j}=\Data_{j+1}$
  and $\Graph_{j+1}$ is a path graph $v_0 v_1\dots v_n v_{n+1}$ with new
  distinct vertex $v_{n+1}$ where $\state(v_{n+1})=v^0_\gamma$.
  We mapped the dispatch rule to a $\push$ of the current state to the pushdown
  and jumping to the new initial state, i.e., we go from $(x_j,\DData_j)$ to
  $(x_{j+1},\DData_{j+1})$ where $\DData_j=\DData_{j+1}$ and
  $x_{j+1}=x_j\cdot s^0_\gamma$. Obviously, $\Graph_{j+1}\triangleright x_{j+1}$.
\item[Test/Assignment:]
  $\Graph_{j+1}$ equals $\Graph_{j}$ except for $\state_j(v_n)=s$ and
  $\state_{j+1}(v_n)=s'$ and a possible change of $\Data_{j+1}$ according
  to the underlying data action. Executing the same action on $\Pp_\Aa$
  assures that $\DData_{j+1}=\Data_{j+1}$ and changing the control state of
  the pushdown only changes $x_j=w\cdot s$ to $x_{j+1}=w\cdot s'$; thus,
  $\Graph_{j+1}\triangleright x_{j+1}$.
\item[Termination:]
      To apply the action $\Graph_j$ consists of a (non-empty) path ending in
      $v$ with $\state_j(v)\in F$ and $\Graph_{j+1}=\Graph_{j}\setminus v$,
      and $\Data_j=\Data_{j+1}$. Note that $\Graph_{j+1}$ could be possibly
      empty.
      Given a $(x_j,\DData_j)$ according to the induction hypothesis, then
      we have to consider two cases: either $x_j=w_j\cdot y_j$ with $w_j\in S^+$
      and $y_j\in S$ (i.e., there is at least one element on the stack), or
      $x_j=y_j\in S$ (i.e., stack is empty). In the second case, we know that
      $x_j\in S_\main$ and by the induction hypothesis, that $x_j=s^0_\main$
      and $\Graph_j$ a path of length 1.
      Now, $\Pp_\Aa$ takes the $\emptystack$ transition leading to the (bottom)
      state $\e$, i.e., $x_{j+1}=\e$, hence $\Graph_{j+1}$ is empty and
      $\Graph_{j+1}\triangleright \e$.
      If the stack is not empty, then we can take a \pop transition
      such that $x_{j+1}=w\in S^+$ for $x_j=w\cdot s$, hence
      $\Graph_{j+1}\triangleright
      x_j$. Obviously $\Data_{j+1}=\Data_j=\DData_j=\DData_{j+1}$.
 \end{description}
  (Recall that we asserted dispatch and scheduling/dequeueing to be atomic, so we
  do not need to consider other actions of the scheduler.)

  The reverse direction follows analogously as the previous inductive
  construction used necessary \emph{sufficient} steps.
  \qed
\end{proof}

\propsyncqdaspds*
\begin{proof}
  We assert that the semantics of $\Pp_\Aa$ is the usual semantics for
  pushdown systems with data, i.e., an infinite transition system with
  configurations $c=(y,w,d)\in Y\times S^* \times \DD^\Xx$.
  Thus, we can interpret configurations also as follows:
  $(x,d)\in S^*\times \DD^\Xx$ with $x=w\cdot y\in S^*\cdot (S\cup\{\e\}$.

  Let $(\Graph,\Data)\in\Reach(\Aa)$ be reachable by a run
  $(\Graph_0,\Data_0)a_1(\Graph_1,\Data_1)a_2\dots a_n(\Graph_n,\Data_n)$.
  Then we can induce a run $(x_0,\DData_0)a_1(x_1,\DData_1)a_2\dots
  a_n(x_n,\DData_n)$ in $\Pp_\Aa$ such that $\Data_i=\DData_i$ and
  $\Graph_i\triangleright x_i$ for $0\leq i \leq n$.

  By construction of $\Pp_\Aa$, $x_0\triangleright G_0$ and $\Data_0=\DData_0$.
  We now assume that there exists a prefix of the \qdas's run of length $0\leq
  j\leq n$ of the form
  $(\Graph_0,\Data_0)\dots
  (\Graph_j,\Data_j)$ such that there exists a run of the pushdown system
  $(x_0,\DData_0)\dots(x_j,\DData_j)$ that fullfills the induction hypothesis.
  We now consider the outcome of a \qdas transition labeled $a_{j+1}$.
  We know that $\Graph_j$ must be a path of vertices $v_0\dots v_n$ connected
  by wait edges.

\begin{description}
\item[Sync. dispatch:] dispatching a block $\gamma$ on queue $q$
  leads to $(\Graph_{j+1},\Data_{j+1})$ with $\Data_{j}=\Data_{j+1}$
  and $\Graph_{j+1}$ is a path graph $v_0 v_1\dots v_n v_{n+1}$ with new
  distinct vertex $v_{n+1}$ where $\state(v_{n+1})=v^0_\gamma$.
  We mapped the dispatch rule to a $\push$ of the current state to the pushdown
  and jumping to the new initial state, i.e., we go from $(x_j,\DData_j)$ to
  $(x_{j+1},\DData_{j+1})$ where $\DData_j=\DData_{j+1}$ and
  $x_{j+1}=x_j\cdot s^0_\gamma$. Obviously, $\Graph_{j+1}\triangleright x_{j+1}$.
\item[Test/Assignment:]
  $\Graph_{j+1}$ equals $\Graph_{j}$ except for $\state_j(v_n)=s$ and
  $\state_{j+1}(v_n)=s'$ and a possible change of $\Data_{j+1}$ according
  to the underlying data action. Executing the same action on $\Pp_\Aa$
  assures that $\DData_{j+1}=\Data_{j+1}$ and changing the control state of
  the pushdown only changes $x_j=w\cdot s$ to $x_{j+1}=w\cdot s'$; thus,
  $\Graph_{j+1}\triangleright x_{j+1}$.
\item[Termination:]
      To apply the action $\Graph_j$ consists of a (non-empty) path ending in
      $v$ with $\state_j(v)\in F$ and $\Graph_{j+1}=\Graph_{j}\setminus v$,
      and $\Data_j=\Data_{j+1}$. Note that $\Graph_{j+1}$ could be possibly
      empty.
      Given a $(x_j,\DData_j)$ according to the induction hypothesis, then
      we have to consider two cases: either $x_j=w_j\cdot y_j$ with $w_j\in S^+$
      and $y_j\in S$ (i.e., there is at least one element on the stack), or
      $x_j=y_j\in S$ (i.e., stack is empty). In the second case, we know that
      $x_j\in S_\main$ and by the induction hypothesis, that $x_j=s^0_\main$
      and $\Graph_j$ a path of length 1.
      Now, $\Pp_\Aa$ takes the $\emptystack$ transition leading to the (bottom)
      state $\e$, i.e., $x_{j+1}=\e$, hence $\Graph_{j+1}$ is empty and
      $\Graph_{j+1}\triangleright \e$.
      If the stack is not empty, then we can take a \pop transition
      such that $x_{j+1}=w\in S^+$ for $x_j=w\cdot s$, hence
      $\Graph_{j+1}\triangleright
      x_j$. Obviously $\Data_{j+1}=\Data_j=\DData_j=\DData_{j+1}$.
 \end{description}
  (Recall that we asserted dispatch and scheduling/dequeueing to be atomic, so we
  do not need to consider other actions of the scheduler.)

  The reverse direction follows analogously as the previous inductive
  construction used necessary \emph{sufficient} steps.
  \qed
\end{proof}

\begin{restatable}{lemma}{lengenparikautomata}\label{lem:genparikautomata}
  Given a finite set $S$ and a function $f:S\rightarrow \NN$,
  then there exists a finite automaton $\Ff_f$ with
  alphabet $S$ of size exponential in $\vert S\vert$ and polynomial in
  (in the binary encoding of)
   $max_{s\in S}f(s)$
  such that $\Ll(\Ff_f)=\{w \in S^* : |w|_s\geq f(s) \text{ for all } s\in
  S\}$.
\end{restatable}
\begin{proof}
  Given a set $S$ and a function $f:S\mapsto \NN$. Let $k=max_{s\in S}f(s)$
  (which must exists as $S$ is finite). Then  $\Ff_f$ is the finite
  automaton $\tuple{Q,S,q^0,\Delta,q^f}$ with
  states $Q=S\times\{0\dots k\}$ (interpreted as an $S$-indexed vector of values
  in $0\dots k$), an action alphabet $S$, the initial state is $q^0$ where
  $q^0(s)=f(s)$, the finial state is $q^f$ where $q^f(s)=0$.
  The transitions of $\Ff_f$ are defined as follows:
  $(q,s,q')\in\Delta$ iff $q'(s)=q(s)-1$ for $q(s)>1$, else $q'(s)=q(s)$,
  and for all $t\in S\setminus\{s\}$ we have $q'(t)=q(t)$. Thus
  each transition labeled by an action $s$ reduces the ``counter'' $q(s)$
  by one until zero and once arrived at zero, the counter $q(s)$
  remains zero for any further $s$ action. Further, the control structure
  of $\Ff_f$ is acyclic (except for the loops at $q^f$), thus
  each run can visit each state in $Q\setminus \{q^f\}$.

  If $w=a_1\dots a_n\in\Ll(\Aa)$ then it was accepted by a run $q_0a_1q_1\dots
  a_nq_n$ where $q_0=q^0$ and $q_n=q^f$. Due to our construction of $\Delta$, it
  holds for $w=a_1\dots a_n$ that $|w|_s\geq q_0(s)=f(s)$ for all $s\in S$.
    If $w\notin \Ll(\Aa)$ then there exists a run $q_0 a_1q_1\dots a_nq_n$
  where $q_0=q^0$ and for $q_n\neq q^f$ it holds that there exists at least one
  $s\in S$ such that $q_n(s)>0$, each transition $(q_{i-1},a_i,q_i)$ assures that
  $q_{i-1}(s)\geq q_i(s)$, hence $|w|_s<f(s)$ for at least one $s\in S$.\qed
\end{proof}

\propsyncqdasparikhsim*
\begin{proof}[Prop.\,\ref{prop:syncqdas_parikh_sim}]
  First, we construct the \pds with data $\Pp_\Aa$ and states $S$
  as mentioned before. Then, we
  translate the \pds with data to a bisimilar \pds without data
  $\widehat{\Pp_\Aa}=\tuple{\widehat{Y},\widehat{y^0},
  \widehat{\Phi},\widehat{\Sigma}, \widehat{\Delta}}$
  by encoding all possible valuations of variables into the {\pds}'s states by the
  standard product construction, i.e.,
  $\widehat{Y}=S\times\left(\Xx\times\DD\right)$. Given $y\in\widehat{Y}$,
  let $S(y)\in S$ denote the original state component.
  Note: $\widehat{\Pp_\Aa}$ is at most
  exponentially larger as $\Pp_\Aa$ and  this construction does not change the
  pushdown system's behaviour with respect to the stack but only internal
  actions.

  Second, from the function $f$, we construct the automaton
  $\Ff_f=\tuple{Q,S,q^0,\Delta_\Ff,q^f}$ analogous to
  Lemma\,\ref{lem:genparikautomata}.

  \newcommand{\Aaf}{{\Aa,f}}
  Finally, we define the \pds
  $\Pp_{\Aaf}=\tuple{Y,y^0,\Phi,\Sigma,\Delta_\Aaf}$
  as follows
  \begin{myitemize}
    \item states are $Y=\widehat{Y}\dotcup Q$ (assuring
       disjointness by relabeling when necessary)
    \item $y^0=\widehat{y^0}$ is the initial state
    \item $\Phi=\widehat{\Phi}$ is the stack alphabet
      (where $\widehat{\Phi}=S$ due to the above construction)
    \item $\Sigma=\widehat{\Sigma} \cup\{\e\}$
    \item a tuple $(y,a,y')$ is a rule in $\Delta_\Aaf\subseteq
      Y\times\Sigma\times Y$ iff one of the following
      holds
      \begin{myitemize}
        \item $(y,a,y')\in\widehat{\Delta}$ (include all transition rules of
          $\widehat{\Pp_\Aa}$);
        \item $a=\pop(s)$ for $s\in\Phi$ and $(q,s,q')\in\Delta_\Ff$ (include
          rules of $\Ff_f$ and change an $s$ action to {\pop(s)} for $s\in S$);
        \item $y\in \widehat{Y}$, $a=\push(z)$ for $z=S(y)$,
          and $y'=q^0$ (connect all states in
          $\widehat{Y}$ with the initial state of $\Ff_f$, additionally stocking
          the current ``state''-component on the stack).
      \end{myitemize}
  \end{myitemize}
  Note that $\Pp_{\Aaf}$ is of size exponential with respect to both the \qdas
  and $f$ due to serial composition.

  We now have to show that if there is a run in $\Pp_\Aaf$ that reaches
  the state $q^f$, then there exists configuration $c=(\Graph,\Data)$ of
  $\Aa$ such that $f\mleq\Parikh(G)$.

  Assert that there exists a run of $\Pp_\Aaf$ reaching $q^f$, then
  it must be of the following form
  $\langle x_0, a_1, x_1, \dots, a_k,
  x_k, a_{k+1}, x_{k+1}, a_{k+2},\dots, a_n, x_n\rangle$
  where $x_i=(y_i,w_i)\subseteq Y\times S^*$ are the corresponding infinite
  transition systems configurations.
  Further, $y_0=y^0$, $y_n=q^f$, $y_{k+1}=q^0$, and $\langle y_1\dots
  y_k\rangle$ is a
  subrun that only uses states in $\widehat{Y}$ as well as transitions in
  $\widehat{\Pp_\Aa}$; $\{y_{k+1},\dots,y_{n}\}\subseteq Q$ and the corresponding
  transitions are derived from $\Delta_\Ff$, as well as $a_{k+1}=\push(S(y_k))$.

  Let us take a closer look on the first part of the run:
  $\langle y_0, a_1,\dots, a_n, x_k\rangle$ is equivalent to a run of
  $\widehat{\Pp_\Aa}$ that reaches a configuration $x_k$. The latter is,
  following Propositions\,\ref{prop:sync_qdas_pds} and \ref{prop:pds_sync_qdas}, similar to a run of the
  original \qdas $\Aa$ that reaches a configuration $c=(\Graph,\Data)$ where
  $\Graph\triangleright y_k\cdot S(y_k)$. Thus, $c\in\Reach(\Aa)$.

  The transition $(x_k, \push(S(y_k)), x_{k+1})$ now transfers the
  encoding of $\Graph$ to the stack, i.e., $w_{k+1}=y_k\cdot S(y_k)$.
  All other information on data encoded in $y_k$ is lost in this step.

  Now, by Lemma\,\ref{lem:genparikautomata} we know that the subrun
  $\langle x_{k+1}, a_{k+2}, \dots, a_n, x_n\rangle$ leading to the final
  state of $\Ff_f$ assures that
  $|w_{k+1}|_s\geq f(s)$ for all  $s\in S$. Hence, for the previously found
   $c=(\Graph,\Data)\in\Reach(\Aa)$ it holds that $f\mleq\Parikh(G)$.\qed
\end{proof}

Let us take a closer look on the dispatches that happen in runs of
synchronous \qdas that have only \emph{serial} queues. Assume a run of
such a \qdas, and suppose the first dispatch performed along this run
(by \texttt{main}) is $\disps(q,\gamma)$. As the dispatch is
synchronous, \texttt{main} is blocked, and the scheduler has to
dequeue $\gamma$ to let the system progress. Cleraly, if $\gamma$
performs a synchronous dispatch $\disps(q,\gamma')$ to the same queue
$q$, we reach a deadlock. Indeed, the task running $\gamma$ is blocked
by the synchronous dispatch of $\gamma'$, but we need to wait for the
termination of $\gamma$ to be able to dequeue $\gamma'$ from $q$
(because $q$ is serial). So, $\gamma$ has to dispatch its blocks to
other queues. For the same reason, we also reach a deadlock if a block
called by $\gamma$ performs a synchronous dispatch into $q$. We
conclude that, in all reachable \ctg, the following holds for all
queues: either the queue contains one block and there is no running
task from this queue, or the queue is empty, and there is at most one
running task from this queue.  Hence, all the reachable \ctg have at
most $\vert SQID\vert +2$ vertices.  Thus, the pushdown systems used
in all previous constructions have bounded stack height and we can
apply the emptiness test on a finite state system when proving
Proposition\,\ref{prop:syncqdas_parikh_sim}.  The lower bound can be
derived from Proposition\,\ref{prop:queueless}. Thus we can derive:

\propsyncqdasparikhsim*

\subsubsection{From \pds to \qdas}

Given a \pds $\Pp$, we construct a synchronous \qdas $\Aa_\Pp$ as shown in
Figure\,\ref{fig:sim_pds}. The underlying idea is the inverse of
the above simulation: we map a \push action of a letter $\phi$ to synchronous
dispatch call of a block $\phi$ and simulate the stack contents in the \ctg such
that we can only map a \pop action to a task's termination if we match the
topmost letter of the stack, encoded in the block name.

\begin{figure}[t!]
  \centering
  \begin{minipage}{0.5\linewidth}
      \begin{lstlisting}
global state := $x^0$
global c_queue q

def $\phi$(): // for each $\phi\in\Phi\dotcup\{\main\}$

  while(true):
    select $(s,a,s')\in\Delta_\Pp$ where state=$s$

    if $a=\push(\phi')$ :
      state := $s'$
      dispatch_s(q,$\phi'$)

    if $a=\pop(\phi)$ and $\phi=\phi'$ :
      state := $s'$
      terminate
      \end{lstlisting}
    \end{minipage}
 \begin{minipage}[t]{0.5\linewidth}
\begin{tikzpicture}
  [zstd/.style={state,font=\tiny},anchor=west]

  \draw (-0.5,1.2) node[anchor=west] (dummy){\ctg in $\Reach(\Aa_\Pp)$:};
  \draw (0,0.5) node[zstd] (00) {$\main$};
  \draw (00)+(1,0) node[zstd] (01) {$\phi_1$};
  \draw (01)+(1,0) node[font=\tiny,anchor=west] (02) {\dots};
  \draw (02)+(0.5,0) node[zstd] (03) {$\phi_k$};
  \draw[dashed] (00) edge[->] (01)
    (01) edge[->] (02)
    (02) edge[->] (03);

    \begin{pgfonlayer}{background}
      \draw node[fit=(00) (03),fill=gray!20] (c1) {};
      \draw (c1.north east)
        node[anchor=south east,inner sep=0pt,font=\tiny]
        {stack};
    \end{pgfonlayer}
\end{tikzpicture}
\end{minipage}

    \caption{From a pushdown system  to a \qdas: \main and
    $\phi$ for $\phi\in\Phi$ \label{fig:sim_pds}}
    \vspace{-4ex}
\end{figure}

The control state of the \pds is stored in the variable \texttt{state} and
the behaviour of the control structure of \Pp is encoded as non-determinstic
choice (line~$7$) that assures that reaching the dispatch and termination actions
(lines~$11$/$15$) demands that the selected transition rule harmonizes with the
current change of the variable \texttt{state} from $s$ to $s'$ and that a
$\push(\phi')$ action is only possible if the currently running task is labeled
by the blockname $\phi'$ (line~$13$).

A reachable configuration of $\Aa_\Pp$
is given by $(\Graph,\Data)$ where $\Graph$ is---as discussed before---a
path of vertices $v_0v_1\dots v_k$. As before, synchronous dispatch calls assure
there is no more than one task
active at the same time. Given $c=(\Graph,\Data)\in\Reach(\Aa_\Pp)$
and a configuration $y=(x,w)\in X\times \Phi^*$ that is reachable in $\Pp$; then
$c$ is represented by $y$, written $c \triangleright y$, iff
$\Data(\texttt{state})=x$ and for $w=w_1\dots w_k$
$\lambda(v_i)=w_i$ for $1\leq i\leq k$ and $\lambda(v_0)=\main$. Hence, the state
of the \pds is stored in the variable \texttt{state}, and the path $v_1\dots v_k$
encodes in the underlying task's blocks the stack content, where the empty stack
is represented by a single vertex labeled by \main.

\begin{restatable}{proposition}{propsyncpdsqdas}
\label{prop:sync_pds_qdas}  Given a pushdown system $\Pp$, then we can generate  a
  synchronous \qdas $\Aa_\Pp$ such that the following holds:
  for any
  run $\pi=y_0a_1y_1\dots a_ny_n$ in $\Pp$ there exists a
  run $\rho=c_0a_1c_1\dots a_nc_n$ of $\Aa_\Pp$ such that
  for all $c_i\triangleright x_i$
  ($0\leq i \leq n$),
  and vice versa.
\end{restatable}
\begin{proof}
  Given a run $\rho=y_0 a_1 y_1 \dots a_k y_k$ of the \pds $\Pp$.
  W.l.o.g. let us consider in the following underlying sequence
  of configurations and fired transition rules
  $y_0 \delta_1 y_1 \dots \delta_k y_k$ where
  $\delta_i=(x_i,a_i,x_i')\in\Delta_\Pp$ for $1\leq i \leq k$.

  We show inductively how $\Aa_\Pp$ generates a run that
  simulates $\rho$.

  For the initial configuration of $\Pp$ $y_0=(x^0,\e)$
  and the initial configuration $c_0=(\Graph,\Data)$
  with $\Graph$ consists of a single node $v_0$ with
  $\lambda(v_0)=\main$ and $\Data(state)=x^0$ it holds
  that $c_0 \triangleright y_0$.

  Now assert that the \pds $\Pp$ reached configuration $y_i$ ($0\leq i \leq k$)
  such that $\Aa_\Pp$ simulated the prefix of the run until
  $c_i=(\Graph_i,\Data_i)$ with
  $c_i\triangleright y_i$. Assert that $\Graph_i$ is a path $v_0 v_1\dots v_l$.
  We do a case-by-case analysis with respect
  to $\delta_{i+1}=(x,a,x')$ that leads to $y_{i+1}$:
  \begin{myitemize}
    \item only the task corresponding to $v_l$ is active and
      the only way to exit its \texttt{while} loop is
      via the lines $11$ and $15$, that assure that line $7$ selected
      $\delta=(x,a,x')\in\Delta_p$ with $\Data_i(state)=x$, and that we
      set $\Data_{i+1}(state)=x'$;;
    \item if $a=\push(\phi)$ for $\phi\in\Phi$, then
      we fire the synchronous dispatch that leads to $\Graph_{i+1}=
      v_0\dots v_l v_{l+1}$ with $\lambda(v_{l+1})=\phi$, thus
      $(\Graph_{i+1},\Data_{i+1})\triangleright y_{i+1}$;
    \item if $a=\pop(\phi)$ for $\phi\in\Phi$ and we left the
      while loop then $\lambda(v_l)= \phi$  (by line $13$), and
      $\Graph_{i+1}$ equals $v_0 \dots v_{l-1}$, thus
      $(\Graph_{i+1},\Data_{i+1})\triangleright y_{i+1}$.
  \end{myitemize}

  The reverse direction follows analogously by considering lines $10;11$
  and $14;15$ as atomic actions (i.e., setting the \texttt{state} variable
  and changing the call graph of the \qdas).
\end{proof}
\subsection{Asynchronous Concurrent \qdas}

\propfromqdastopn*

The proof of the proposition relies on the following lemma, showing
that $N_\Aa$ can simulate precisely the sequence of Parikh images that
are reachable in $\Aa$. Let $(G,\Data)$ be a configuration of $\Aa$,
and let $m$ be marking of $N_\Aa$. We say that \emph{$m$ encodes
  $(G,\Data)$}, written $m\rhd (G,\Data)$ iff: $(i)$ for all
$x\in\Xx$: $m(x,\Data(x))=1$, $(ii)$ for all $x\in\Xx$: for all
$d\in\DD\setminus\{\Data(x)\}$: $m(x,d)=0$ and $(iii)$ for all $s\in
S$ $m(s)=\Parikh(G)(s)$. Then:

\begin{lemma}\label{lemma:from-qdas-to-pn}
  Let $\Aa$ be a concurrent asynchronous \qdas with set of variables
  $\Xx$ and set of locations $S$, and let $N_\Aa$ be its associated
  \pn.  Then, for all $(G,\Data)\in\Reach(\Aa)$ there is
  $m\in\Reach(N_\Aa)$ s.t. $m\rhd (G,\Data)$ and for all
  $m\in\Reach(N_\Aa)$, there is $(G,\Data)\in\Reach(\Aa) $ s.t. $m\rhd
  (G,\Data)$.
\end{lemma}
\begin{proof}
  We prove the two statements separately.

  Let $(G,\Data)$ be a configuration in $\Reach(\Aa_N)$, and let
  $(G_0,\Data_0)a_0(G_1,\Data_1)a_1\cdots a_{n-1}(G_n,\Data_n)$ be a
  run s.t. $(G,\Data)=(G_n,\Data_n)$. Let us build, inductively, a run
  $m_0m_1\cdots m_k$ of $N_\Aa$ s.t. $m_k\rhd (G,\Data)$. The
  induction is on the length $n$ of the \qdas run.

  \textbf{Base case $n=0$.} It is easy to check that $m_0\rhd
  (G_0,\Data_0)$.

  \textbf{Inductive case $n=\ell$.} Let us assume that $m_0m_1\cdots
  m_j$ is a run of $N_\Aa$ s.t. $m_j\rhd (G_{\ell-1},\Data_{\ell-1})$,
  and let us show how to complete it, if needed. We consider several
  case depending on $a_{n-1}$. In the case where $a_{n-1}=\varepsilon$
  and the scheduler action consists in dequeueing a block from a queue,
  we have $\Parikh(G_{\ell-1})=\Parikh(G_\ell)$ and
  $\Data_\ell=\Data_{\ell-1}$. By induction hypothesis $m_j\rhd
  (G_{\ell-1},\Data_{\ell-1})$, hence $m_j\rhd (G_\ell, \Data_\ell)$,
  and we do not add elements to the run built so far. In the case
  where $a_{\ell-1}=\dispa(\gamma,q)$, we assume
  $(s,a_{\ell-1},s')\in\Delta$ is the corresponding \lts
  transition. Clearly,
  $\Parikh(G_\ell)(s')=\Parikh(G_{\ell-1})(s')+1$,
  $\Parikh(G_\ell)(s)=\Parikh(G_{\ell-1})(s)-1$,
  $\Parikh(G_\ell)(s^0_\gamma)=\Parikh(G_{\ell-1})(s^0_\gamma)+1$ and
  for all other location $s$:
  $\Parikh(G_\ell)(s)=\Parikh(G_{\ell-1})(s)$. It is easy to check
  that the \pn transition $t$ s.t. $I(t)(p)=1$ iff $p=s$ and
  $O(t)(p)=1$ iff $p\in\{s',s^0_\gamma\}$ is fireable from $m_j$ (as
  $m_j\rhd (G_{\ell-1},\Data_{\ell-1})$ by induction hypothesis) and
  yields the same effect, i.e. the marking $m$ with
  $m_j\xrightarrow{t}m$ is s.t. $m\rhd (G_\ell,\Data_\ell)$. All the
  other cases (test, assignment and task termination) are treated
  similarly.  \medskip

  Now, let $m_0m_1\cdots m_n$ be a run of $N_\Aa$ and let us build,
  inductively, a run $(G_0,\Data_0)a_0(G_1,\Data_1)a_1\cdots
  a_{k-1}(G_k,\Data_k)$ s.t. $m_n\rhd (G_k,\Data_k)$ \emph{and} all
  the queues are empty in $G_n$. The induction is on the length $n$ of
  the \pn run.

  \textbf{Base case $n=0$.} It is easy to check that $m_0\rhd
  (G_0,\Data_0)$.

  \textbf{Inductive case $n=\ell$.} Let us assume that
  $(G_0,\Data_0)a_0\cdots a_{j-1}(G_j,\Data_j)$ is a run of $\Aa$
  s.t. $m_{\ell-1}\rhd (G_j,\Data_j)$ and all the queues are empty in
  $G_j$. Let $t$ be the \pn transition
  s.t. $m_{\ell-1}\xrightarrow{t}m_\ell$ and let us show how we can
  extend the run of $\Aa$. We consider several cases. If $t$ is a
  transition that corresponds to an asynchronous dispatch, then there
  are $s$, $s'$, $\gamma$ and $q$ s.t. $I_t(p)=1$ iff $p=s$ and
  $O_t(p)=1$ iff $p\in\{s', s^0_\gamma\}$. By definition of $N_\Aa$,
  there is a transition $(s,\dispa(\gamma,q),s')$ in $\Aa$. Moreover,
  $m_{\ell-1}(s)\geq 1$, since $t$ is fireable from $m_{\ell-1}$. As
  $m_{\ell-1}\rhd (G_j,\Data_j)$, the $(s,\dispa(\gamma,q),s')$ is
  fireable from $(G_j,\Data_j)$, and leads to a configuration
  $(G_{j+1}, \Data_{j+1})$, where a $\gamma$ block has been enqueued
  in $q$, hence $\Data_{j+1}=\Data_j$,
  $\Parikh(G_{j+1})(s)=\Parikh(G_j)(s)-1$,
  $\Parikh(G_{j+1})(s')=\Parikh(G_j)(s')+1$,
  $\Parikh(G_{j+1})(s^0_\gamma)=\Parikh(G_j)(s^0_\gamma)+1$ and for
  all other state $s''$: $\Parikh(G_{j+1})(s'')=\Parikh(G_j)(s'')$. It
  is easy to check that $m_\ell\rhd (G_{j+1},\Data_{j+1})$, however,
  queue $q$ contains a call to $\gamma$ in $G_{j+1}$ and is thus the
  only non-empty queue in this \ctg. Thus, from
  $(G_{j+1},\Data_{j+1})$, we execute the scheduler action that
  dequeues from $q$. This has no effect on the Parikh image of the
  \ctg. Thus, we reach $(G_{j+2},\Data_{j+2})$
  s.t. $\Data_{j+1}=\Data_{j+2}$, $\Parikh(G_{j+1})=\Parikh(G_{j+2})$,
  hence $m\ell\rhd (G_{j+2},\Data_{j+2})$ too, and all the queues are
  empty in $G_{j+2}$, which concludes the induction step. All the
  other cases are treated similarly.\qed
\end{proof}

We can now prove Proposition~\ref{prop:from-qdas-to-pn}:
\begin{proof}
  It is easy to check that the construction of $N_\Aa$, as described
  above, is polynomial. Then, assume $f$ is Parikh coverable in $\Aa$,
  i.e. there is $(G,\Data)\in \Reach(\Aa)$
  s.t. $f\mleq \Parikh(G)$. By Lemma~\ref{lemma:from-qdas-to-pn},
  there is $m'\in\Reach(N_\Aa)$ s.t. $m'\rhd (G,\Data)$. Hence, for
  all $s\in S$: $m'(s)=\Parikh(G)(s)$. So, for all $s\in S$:
  $m(s)=f(s)\leq \Parikh(G)(s)=m'(s)$. Hence, $m\mleq m'$ (as $m(p)=0$
  for all $p\not\in S$). Since $m'\in\Reach(N_\Aa)$, we conclude that
  $m\in\Cover(N_\Aa)$. On the other hand, assume $m\in\Cover(N_\Aa)$,
  with $m(p)=0$ for all $p\not\in S$, and let $f$ be s.t. for all
  $s\in S$: $f(s)=m(s)$. Since $m\in\Cover(N_\Aa)$, there is
  $m'\in\Reach(N_\Aa)$ s.t. $m\mleq m'$. By
  Lemma~\ref{lemma:from-qdas-to-pn}, there is
  $(G,\Data)\in\Reach(\Aa)$ s.t. $m'\rhd (G,\Data)$. Thus, by
  definition of $\rhd$, for all $s\in S$: $m'(s)=\Parikh(G)(s)$. Thus,
  since $m\mleq m'$ and by definition of $f$, we conclude that for all
  $s\in S$: $f(s)=m(s)\leq m'(s)=\Parikh(G)(s)$. Hence, $f$ is
  Parikh-coverable in $\Aa$.\qed
\end{proof}

\propfrompntoqdas*

The proof of Proposition~\ref{prop:from-pn-to-qdas} is split into two
lemmata, given hereunder. They rely on an alternate characterization of
$\Cover(N)$. That is, $m\in\Cover(N)$ iff $m$ is reachable by a
so-called \emph{lossy} run of $N$, i.e. a sequence of markings
$m_0'm_1'\cdots m_n'$ s.t. $m_0'\mleq m_0$ and for all $0\leq i\leq
n-1$: there is $\mbar_{i+1}$ and a transition $t_i$
s.t. $m_i'\xrightarrow{t_i}\mbar_{i+1}$ and
$m_{i+1}'\mleq\mbar_{i+1}$. Intuitively, a lossy run corresponds to
firing a transition of the PN, and then spontaneously losing some
tokens. The proof of these lemmata also assumes that each $p\in P$,
the \lts $\Ts_p=\tuple{\{s^0_P, s^{mid}_p,s^{fin}_p\}, s^0_p,
  \Sigma,\Rightarrow}$ is as depicted in Fig.~\ref{fig:lts-p}.

\begin{figure}
  \centering
  \includegraphics[scale=.5]{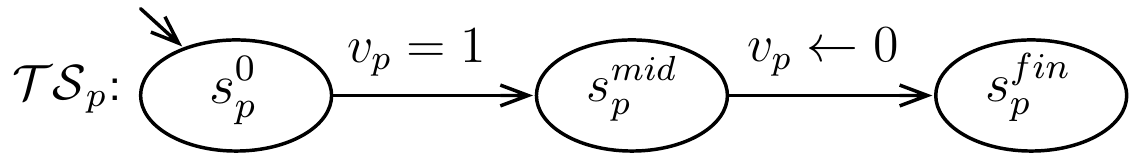}
  \caption{The \lts of bloc {\tt p}.}
  \label{fig:lts-p}
\end{figure}

\begin{lemma}
  Let $N=\tuple{P,T,m_0}$ be a \pn, and let $\Aa_N=\tuple{\CQID,
    \emptyset, \Gamma, \mathtt{main}, \Xx, \Sigma,
    (\Ts_\gamma)_{\gamma\in\Gamma}}$ be its corresponding \qdas.
  \textbf{If} $m\in\Cover(N)$ \textbf{then} there exists
  $(G,\Data)\in\Reach(\Aa_N)$ s.t. $G\rhd m$.
\end{lemma}
\begin{proof}
  Let $m$ be a marking from $\Cover(N)$.  and let $m_0'm_1'\cdots m_n'$
  be a lossy \pn run s.t. $m=m_n$. The proof is by induction on
  the length of the run. More precisely, we show that, for all $0\leq
  i\leq n$, there is a reachable configuration
  $(G_i,\Data_i)\in\Reach(\Aa_N)$ s.t.: for all $p\in P$:
  $\Data_i(v_p)=0$, $G_i=\tuple{V^i, E^i,
    \lambda^i,\queue^i,\state^i}$, $G_i\rhd m$, $m\mleq m_i'$ and
  $E^i=\emptyset$.

  \textbf{Base case: $m_0'$}. Let us consider the run of $\Aa_N$ that
  consists in: $(a)$ executing block \texttt{main} up to line 8, then
  $(b)$ emptying the queue $C$. The execution of $(a)$ has the effect
  that: $(i)$ all $v_p$ variables are initialized to $0$ and keep this
  value, $(ii)$ for all place $p$: \emph{at most} $m_0(p)$ copies of
  block \texttt{p} are asynchronously dispatched in queue $C$ and
  $(iii)$ one copy of block \texttt{trans} is dispatched in $C$. Then,
  the execution of $(b)$ creates one running task for each block that
  is present in $C$. Thus, the execution of $(a)$ followed by $(b)$
  reaches a configuration $(G_0,\Data_0)$ with
  $G_0=\tuple{V^0=V_T^0\dotcup V_C^0,E^0,\lambda^0, \queue^0,
    \state^0}$ s.t. $V_C^0=\emptyset$ (the queue has been emptied),
  for all $p$: $|\set{v\in V_T^0}{\lambda(v)=\mathtt{p}}|=m_0(p)$,
  $|\set{v\in V_T^0}{\lambda(v)=\mathtt{trans}}|=1$ and
  $E^0=\emptyset$ (the queue is empty and all the calls are
  asynchronous). Moreover, $\state$ is such that each task running a
  \texttt{p} block is still in its initial state $s^0_p$, hence
  $G_0\rhd m_0$. Similarly, the task running the \texttt{trans()}
  block is about to enter the \texttt{while} loop at line 14. Finally,
  as the variables have been initialized to $0$ and not modified, we
  have $\Data_0(v_p)=0$ for all $p\in P$.

  \textbf{Inductive case: $m_i$} Let us assume there exist
  $(G_{i-1},\Data_{i-1})\in\Reach(\Aa_N)$ that respects all the
  conditions given at the beginning of the proof (in particular
  $G_{i-1}\rhd m_{i-1}$). Let $t_i$ and $\mbar_i$ be the \pn transition
  and marking s.t. \smash{$m_{i-1}\xrightarrow{t_i}\mbar_i$} and
  $m_i\mleq\mbar_i$ and let us show that $\Aa_N$ can simulate it. This
  is achieved by the following sequence of actions in $\Aa_N$. First,
  the block executing \texttt{trans} enters the \texttt{while} loop at
  line 14 and selects $t_i$ as transition $t$. Then, it sets all the
  variables $v_p$ s.t. $I_{t_i}(p)=1$ to $1$. Thus, at that point
  $v_p$ contains $1$ iff $I_{t_i}(p)=1$, since all $v_p$ variables
  were equal to $0$ by induction hypothesis. Then, the task executing
  \texttt{trans} is blocked as it need to wait up to the point were
  all $v_p$ are equal to $0$. Since $G_{i-1}\rhd m_{i-1}$ by induction
  hypothesis, we know that there are, in $G_{i-1}$, $m_{i-1}(p)$ tasks
  executing block $\mathtt{p}$, for all $p\in P$. However, $t_i$ is
  fireable from $m_{i-1}$, and a loss of $\mbar_i-m_i$ token is still
  possible after the firing. Hence, $m_{i-1}(p)\geq
  (I_{t_i}(p)+\mbar_i(p)-m_i(p))$ for all $p$. Thus, for all $p$,
  there is at least $(I_{t_i}(p)+\mbar_i(p)-m_i(p))$ tasks executing
  $\mathtt{p}$ in $G_{i-1}$. Thus, we complete the run of $\Aa_N$ by
  letting, for all $p$, $(I_{t_i}(p)+\mbar_i(p)-m_i(p))$ \texttt{p}
  task execute lines 11 in turn one after the other. Then, letting
  them all execute line 12, and reach their final state (Remark that
  all the \texttt{p} task must first execute line 11 before one of
  them can execute line 12, as this sets $v_p$ to $0$ and would
  prevent other tasks to execute line 11). This is possible because
  none of those tasks are blocked, since the \ctg contains no edge, by
  induction hypothesis. At that point, $\Aa_N$ has reached a
  configuration $(G',\Data')$ s.t. $\Data'(v_p)=0$ for all $p\in P$
  (by line 12) and where $G'\rhd
  m_{i-1}-(I_{t_i}+\mbar_i-m_i)$. Moreover, $G'$ still respects all
  the other hypothesis as no new dispatch have been performed. Then,
  the simulation of $t_i$ proceeds by letting the \texttt{trans} task
  finish the current iteration of the main \texttt{while} loop. This
  consists in executing the \texttt{for} loop of line 19, which
  dispatches one \texttt{p} block in $C$ iff $O_{t_i}(p)=1$, i.e., the
  effect of $t_i$ is to add a token to $p$. Finally, the scheduler
  empties queue $C$ and creates tasks for all the blocks that have
  just been added to $C$. It also kills all the \texttt{p} tasks that
  have reached their final state. As a consequence, the configuration
  that is reached is $(G_i, \Data_i)$, where $G_i\rhd
  m_{i-1}-(I_{t_i}+\mbar_i-m_i)+O_{t_i} =
  (m_{i-1}-I_{t_i}+O_{t_i})-\mbar_i+m_i = \mbar_i-\mbar_i+m_i = m_i$
  and $\Data_i$ is s.t. $\Data_i(v_p)=0$ for all $p\in P$. Moreover,
  since the queue has been emptied by the scheduler, $G_i$ contains
  only task nodes and no edge, as all the calls are asynchronous. The
  task executing \texttt{trans} is still active and at line 14, and
  all the \texttt{p} tasks are in their initial state.\qed
\end{proof}

\begin{lemma}
  Let $N=\tuple{P,T,m_0}$ be a \pn, and let $\Aa_N=\tuple{\CQID,
    \emptyset, \Gamma, \mathtt{main}, \Xx, \Sigma,
    (\Ts_\gamma)_{\gamma\in\Gamma}}$ be its corresponding \qdas.
  \textbf{If} there are $(G,\Data)\in\Reach(\Aa_N)$ and $m$
  s.t. $G\rhd m$ \textbf{then} $m(G)\in\Reachloss(N)$.
\end{lemma}
\begin{proof}
  For a \ctg $G$ of $\Aa_N$ with set of vertices $V$, we denote by
  $M(G)$ the marking of $N$ s.t. for all $p\in P$: $M(G)(p)=|\set{v\in
    V}{\state(v)=s^0_p}|$. Thus, in the case where $G$ encodes a
  configuration s.t. \texttt{trans} is at line 14, \texttt{main} is at
  line 8, and all the \texttt{p} blocks are in their initial state,
  then $G\rhd M(G)$.

  In order to establish the lemma, we prove a stronger statement:
  every time we reach, along a run, a configuration $(G,\Data)$
  s.t. \texttt{trans} is at line 14, then $M(G)\in\Reachloss(N)$.
  Formally, let
  $\rho=(G_0,\Data_0)a_0(G_1,\Data_1)a_1(G_2,\Data_2)\cdots(G_n,\Data_n)$
  be a run of $\Aa_N$, where, for all $0,\leq i\leq n$:
  $G_i=\tuple{V_i,E_i,\lambda_i,\queue_i,\state_i}$. Let
  $\pi:\{0,\ldots, k\}\rightarrow \{0,\ldots,n\}$ be the monotonically
  increasing function s.t. $k\leq n$ and for all $0\leq j\leq n$:
  there exists $v\in V_j$ with $\state_i(v)=\mystate{14}{trans}$ iff
  there is $0\leq \ell\leq k$ with $k=\pi(\ell)$.  That is the
  sequence $\pi(1), \pi(2),\ldots, \pi(k)$ identifies the indexes of
  all the configurations of the run where \texttt{trans} is at line
  14. Let us show, by induction on $i$ that all the $M(G_{\pi(i)})$'s
  are reachable \emph{in the lossy semantics} of $N$.

  \textbf{Base case $i=0$} Let us show that $M(G_{\pi(0)})=m_0$, i.e.,
  that the first time \texttt{trans} reaches line 14, $M(G_{\pi(0)})$
  is the initial marking of $N$. Observe that the prefix of the run
  must have the following form. Initially, only the \texttt{main} block
  is executing: it first sets all the variables $v_p$ to $0$, then
  dispatches asynchronously at most $m_0(p)$ calls to each \texttt{p}
  block (for all $p\in P$), then finally dispatches an asynchronous
  call to \texttt{trans} and reaches line 8. Along this execution, the
  scheduler might decide to pick up some \texttt{p} blocks from
  $C$. However, as long as the scheduler has not scheduled the call to
  \texttt{trans}, the \ctg met along the run do not encode any
  marking, by definition of $\rhd$. When the scheduler starts a task
  to run the \texttt{trans} block, we thus reach a configuration
  $(G,\Data)$ where: $(i)$ the queue $C$ is empty, as dequeueing the
  \texttt{trans} block is possible only if all the \texttt{p} blocks
  have been dequeued, and no other dispatch has been performed; $(ii)$
  all the \texttt{p} tasks are blocked in their initial state as
  $\Data(v_p)=0$ for all $p\in P$; and $(iii)$ \texttt{main} is still
  blocked in the infinite loop at line $8$. Since the scheduler has
  just dequeued \texttt{trans} from $C$, $G$ is necessarily the first
  \ctg to encode a marking, so $G=G_{\pi(0)}$. Moreover, by the loop
  at line 4, it is clear that $G\rhd m$ with $m\mleq m_0$.

  \textbf{Inductive case $i=\ell\geq 1$} The induction hypothesis is
  that $M(G_{\pi(i-1})\in \Reachloss(N)$. Let us consider the
  $\rho'=(G_{\pi(\ell-1)},\Data_{\pi(\ell-1)})\cdots(G_{\pi(\ell)},\Data_{\pi(\ell)})
  $, i.e. the portion of $\rho$ that allows to reach
  $(G_{\pi(\ell)},\Data_{\pi(\ell)})$ from
  $(G_{\pi(\ell-1)},\Data_{\pi(\ell-1)})$. We consider two cases:
  \begin{enumerate}
  \item Either \textbf{trans} has not performed an iteration of its
    main \texttt{while} loop along $\rho'$. In this case, the only
    actions that can occur along $\rho'$ are scheduler actions
    consisting in dequeueing \texttt{p} blocks or the termination of
    some \texttt{p} tasks that where still in state $s^{mid}_p$. In
    both cases, this does not modify the value of $M(G)$, so
    $M(G_{\pi(i)})=M(G_{\pi(i-1})\in \Reachloss(N)$.
  \item Or \texttt{trans} has performed a complete iteration of its
    main \texttt{while} loop possibly interleaved with the dequeue
    of \texttt{p} blocks and the termination of \texttt{p}
    tasks. Since the dequeues and terminations have no influence on
    the value of $M(G)$ as argued above, let us focus on the effect
    of executing one iteration of the \texttt{while} loop. The
    iteration first selects a \pn transition $t$ and sets all the
    variables $v_p$ s.t. $I_t(p)=1$ to $1$. The reached configuration
    is then $(G,\Data)$ where $M(G)=M(G_{\pi(i-1)})$, as these
    operations do not manipulate \texttt{p} blocks or tasks. Then,
    \texttt{trans} is blocked by the test at line 18. As only
    \texttt{p} blocks can set $v_p$ variables to $0$, we are sure
    that, when \texttt{trans} reaches line 19, \emph{at least}
    $I_t(p)$ \texttt{p} blocks have left their initial state, for all
    $p\in P$. Thus, when \texttt{trans} is at line 19, the
    configuration is $(G',\Data')$, where for all $p\in P$:
    $M(G')(p)\leq M(G)(p)-I_t(p)=M(G_{\pi(i-1)})-I_t(p)$. Afterwards,
    \texttt{trans} terminates the iteration of the \texttt{while} loop
    by dispatching $O_t(p)$ \texttt{p} blocks for all $p\in P$, and
    reaches line 14, which finishes $\rho'$. Hence, we reach
    $(G_{\pi(i)},\Data_{\pi(i)})$, where for all $p\in P$:
    $M(G_{\pi(i)})(p)\leq M(G_{\pi(i-1)})-I_t(p)+O_t(p)$. Since
    $M(G_{\pi(i-1)})\in\Reachloss(N)$ by induction hypothesis, we
    conclude that $M(G_{\pi(i)})\in\Reachloss(N)$ too.\qed
  \end{enumerate}

\end{proof}

\subsection{Asynchronous Serial \qdas}
 We
establish the undecidability for asynchronous serial \qdas  by a reduction from the control-state  reachability
problem in a \fifo system. Let $F=\tuple{S_F,s_F^0,M,\Delta_F}$ be a
\fifo system and let $c\in S_F$ be a control state whose reachability
has to be tested. We build the asynchronous serial \qdas
$\Aa_F=\tuple{ \emptyset, \{q\}, \Gamma, \mathtt{main}, \Xx, \Sigma,
  (\Ts_\gamma)_{\gamma\in\Gamma}}$ on domain $\DD=M\cup
S_F\cup\{\e\}$, where $\Gamma=M\cup\{\e,\mathtt{main}\}$,
$\Xx=\{\mathtt{state},\mathtt{head}\}$ and the $\Ts_{\gamma}$ are
given by the pseudo code in Fig.~\ref{fig:fifo-to-qdas}.
\begin{figure}[t]
  \centering
    \begin{minipage}[t]{.45\linewidth}
      \begin{lstlisting}
global state, head
global s_queue q

def main():
  state := $s^0_F$
  head := $\e$
  dispatch_a(q, $\e$)
  while(true): do nothing
      \end{lstlisting}

      \begin{minipage}{0.8\textwidth}
        \begin{tikzpicture}[overlay,remember picture]
          \draw node[anchor=north west,draw,
            rectangle callout,callout relative
            pointer={(1.2,-0.08)},font=\scriptsize, fill=gray!5,text
            width=4.7cm]
          {
      \scriptsize
      Note that the reachability of
      a state $c$ of the \fifo system is explicitely coded
      into the control
      structure.};

        \end{tikzpicture}
    \end{minipage}

    \end{minipage}
    \begin{minipage}[t]{.5\linewidth}
      \begin{lstlisting}[firstnumber=last]
def m(): //for all $m\in M\cup\{\e\}$
  if (head$\neq m$): goto 20
  while(true):
    if (state = $c$): goto 21
    select $(s,a,s')\in\Delta_y$
    if ($s\neq$state): goto 20
    state := $s'$
    if ($a=!n$): dispatch_a(q, $n$)
    else if ($a=?n$):
      head := $n$
      terminate
  while(true): do nothing // wrong guess
  while(true): do nothing // $c$ is reached
      \end{lstlisting}
    \end{minipage}

    \begin{tikzpicture}
  [zstd/.style={state,font=\tiny},zstd2/.style={zstd,rectangle},anchor=west]

  \draw (-0.5,0.6) node[anchor=west] (dummy){\small \ctg type (a):};
  \draw (0,0) node[zstd2] (00) {$m_1$};
  \draw (00)+(0.7,0) node[zstd2] (01) {$m_2$};
  \draw (01)+(0.7,0) node[font=\tiny,anchor=west] (02) {\dots};
  \draw (02)+(0.7,0) node[zstd2] (03) {$m_n$};
  \draw (03)+(0.7,0) node[zstd] (04) {\main};
  \draw (00) edge[->] (01)
    (01) edge[->] (02)
    (02) edge[->] (03);

    \begin{pgfonlayer}{background}
      \draw node[fit=(00) (03),fill=gray!20] (c1) {};
      \draw (c1.north east)
        node[anchor=south east,inner sep=0pt,font=\tiny]
        {queue $q$};
    \end{pgfonlayer}
  \end{tikzpicture}\hspace{1cm}
\begin{tikzpicture}
  [zstd/.style={state,font=\tiny},zstd2/.style={zstd,rectangle},anchor=west]

  \draw (-0.5,0.6) node[anchor=west] (dummy){\small \ctg type (b):};
  \draw (0,0) node[zstd2] (00) {$m_1$};
  \draw (00)+(0.7,0) node[zstd2] (01) {$m_2$};
  \draw (01)+(0.7,0) node[font=\tiny,anchor=west] (02) {\dots};
  \draw (02)+(0.7,0) node[zstd2] (03) {$m_n$};
  \draw (03)+(0.7,0) node[zstd] (04) {$m$};
  \draw (04)+(0.4,0) node[zstd] (05) {\main};
  \draw (00) edge[->] (01)
    (01) edge[->] (02)
    (02) edge[->] (03)
    (03) edge[->,dashed] (04);

    \begin{pgfonlayer}{background}
      \draw node[fit=(00) (03),fill=gray!20] (c1) {};
      \draw (c1.north east)
        node[anchor=south east,inner sep=0pt,font=\tiny]
        {queue $q$};
    \end{pgfonlayer}
\end{tikzpicture}
    \caption{Fifo system encoding into a serial asynchronous
      \qdas / two types of \ctg in this case}
  \label{fig:fifo-to-qdas}
  \label{fig:shapes-aaf}
\vspace{-2ex}
\end{figure}

Intuitively, runs of $\Aa_F$ simulate the runs of $F$, by encoding the
current state of $F$ in variable \texttt{state} and the content of
$F$'s queue into the content of the serial queue \texttt{q}. More
precisely, it easy to check that, once \texttt{main} has reached
line~8, all the \ctg that are reached in $\Aa_F$ are of either shapes
depicted in Fig.~\ref{fig:shapes-aaf}, for $\{m_1,\ldots,
m_n,m\}\subseteq M\cup\{\e\}$. That is, there are at most two running
tasks: \texttt{main} and possibly one task running a $m$ block (for
$\texttt{m}\in M\cup\{\e\}$), that has to terminate to allow a further
dequeue from \texttt{q}. This is because \texttt{q} is a serial queue
and all the dispatches are asynchronous. When the \ctg is of shape
\textsf{(b)}, the duty of the running $m$ block is to simulate a run
of $F$. It runs an infinite \texttt{while} loop (line 11 onwards --
ignore the test at line 10 for the moment), that $(i)$ tests whether
$c$ has been reached (line 12) and jumps to line 20 if it is the case;
$(ii)$ guesses a transition $(s,a,s')$ of $F$; and $(iii)$ checks that
the guessed transition is indeed fireable from the current
configuration of $F$, and, if yes, simulate it. This consists in,
first testing that $s$ is the current state (line 14). If not, the
block jumps to the infinite loop of line 19, which ends the
simulation. Otherwise, the current state is update to $s'$, and the
channel operation is then simulated. A send of message $m$ is
simulated (line 16) by an asynchronous dispatch of block $m$ to
\texttt{q}. The simulation of a receive of $m$ from \texttt{q} is more
involved, as only the scheduler can decide to dequeue a block from
\texttt{q}, and this can happen only if the current running block
terminates (line 19). Still, we have to check that message $m$ is
indeed in the head of \texttt{q}. This is achieved by setting global
variable \texttt{head} to $m$, and letting the next dequeues block
check that itself encodes the value stored into \texttt{head}. This is
performed at line 10. If this test is not satisfied, the block jumps
to the infinite loop of line 20, and the simulation ends. Otherwise,
it proceeds with the simulation. Thus, in all reachable configurations
of $\Aa_F$, a block $m$ (with $m\in M\cup\{\e\}$ will reach line 21
iff $c$ is reachable in $F$. This effectively reduces the control
location reachability of \fifo systems to the Parikh coverability
problem of serial asynchronous \qdas.

The proof of Theorem~\ref{the:async-seri-undec} relies on the
next Lemma, that formalizes the relationship between reachable
configurations of $\Aa_F$ and reachable configurations of~$F$.

For all $\gamma\in\Gamma$, we denote by $\mystate{\ell}{\gamma}$ the
location of $\Ts_\gamma$ that corresponds to line $\ell$ in
Fig.~\ref{fig:fifo-to-qdas}. Then, we say that a configuration
$(\Graph,\Data)$ of $\Aa_F$ encodes a configuration $(s,w)$ of $F$,
written $(\Graph,\Data)\rhd(s,w)$ iff: $(i)$
$s=\Data(\mathtt{state})$, $(ii)$ $\Graph$ is of either shapes in
Fig.~\ref{fig:shapes-aaf} with $w=m_0m_1\cdots m_n$, $(iii)$
$\Parikh(G)(\mystate{8}{{\tt main}})=1$ and $(iv)$ there exists $m\in
M\cup\{\e\}$ s.t. $\Parikh(G)(\mystate{12}{m})=1$. That is, $s$ and
$w$ are encoded as described above, \texttt{main} is at line 8, and
the running $\mathtt{m}$ block is at line 12. Then:

\begin{lemma}\label{lem:from-fifo-to-qdas}
  Let $F$ be a FIFO system, let $c$ be a configuration of $F$, and let
  $\Aa_F$ be its associated \qdas. For all run
  $(s_0,w_0)(s_1,w_1)\cdots(s_n,w_n)$ of $F$ s.t. for all $0\leq i<n$:
  $s_i\neq c$, there exists $(G,\Data)\in\Reach(\Aa_F)$ s.t.
  $(G,\Data)\rhd(s_n,w_n)$. Moreover, for all
  $(G,\Data)\in\Reach(\Aa_F)$ and for all configuration $(s,w)$ of
  $F$: $(G,\Data)\rhd(s,w)$ implies $(s,w)\in\Reach(F)$
\end{lemma}
\begin{proof}
  First, we consider a run $(s_0,w_0)(s_1,w_1)\cdots(s_n,w_n)$ of $F$
  s.t. for all $0\leq i<n$: $s_i\neq c$, and build a run
  $(\Graph_0,\Data_0)a_0(\Graph_1,\Data_1)a_1\cdots
  a_{k-1}(\Graph_k,\Data_k)$ of $\Aa_F$
  s.t. $(\Graph_k,\Data_k)\rhd(s_n,w_n)$, by induction on the length
  of $F$'s run.

  \textbf{Base case $n=0$:} Consider the run of $\Aa_F$ that consists
  in executing lines 5, 6, 7 of \texttt{main} (which sets the
  \texttt{head} variable to $\e$), then dequeueing the $\e$ block from
  the queue, then executing lines 10 and 11 of $\e$. Remark that the
  test at line 10 is not satisfied, as $\mathtt{head}=\e$, and that
  the queue is now empty. Clearly, the resulting configuration
  $(\Graph,\Data)\rhd(s^0_F,w_0)$ as $w_0=\e$.

  \textbf{Inductive case $n=\ell$}. Let us assume that there is a
  reachable configuration $(\Graph, \Data)$ of $\Aa_F$
  s.t. $(\Graph,\Data)\rhd (s_{\ell-1}, w_{\ell-1})$, and let us build
  a sequence of $\Aa_F$ transitions that is fireable from
  $(\Graph,\Data)$ and reaches a configuration encoding
  $(s_\ell,w_\ell)$. In $(\Graph,\Data)$, there is, by definition of
  $\rhd$, a task running a $b$ block, for $b\in M\cup\{\e\}$, that is
  at line 12. Moreover, $\Data(\mathtt{state})=s_{\ell-1}$. Let
  $\delta$ be the transition of $F$
  s.t. $(s_{\ell-1},w_{\ell-1})\xrightarrow{\delta}(s_\ell,w_\ell)$. By
  hypothesis, $s_{\ell-1}\neq c$, hence, we let $b$ execute line 12;
  select $\delta=(s_{\ell-1},a,s_\ell)$ at line 13; execute line 14,
  where the condition of the \texttt{if} is not satisfied as
  $s=s_{\ell-1}=\mathtt{state}$; and execute line 16, which reaches a
  configuration $(\Graph',\Data')$ where
  $\Data'(\mathtt{state})=s_\ell$. We consider three cases to complete
  the simulation of $\delta$ in $\Aa_F$. If $a=!n$, the $b$ task
  performs an asynchronous dispatch of $n$ to \texttt{q}, and jumps to
  line 11, then 12. Clearly, the resulting configuration
  $(\Graph'',\Data'')$ is s.t. $(\Graph'',\Data'')\rhd(s_\ell,w_\ell)$
  (in particular, the dispatch has correctly updated the content of
  the queue). If $a=\e$, the $b$ tasks jumps directly to line 11, then
  to line 12. Again, the resulting configuration $(\Graph'',\Data'')$
  is s.t. $(\Graph'',\Data'')\rhd(s_\ell,w_\ell)$, as the content of
  the queue has not been modified. Finally, if $a=!n$, the running $b$
  block sets \texttt{head} to $n$ and terminates. Let
  $(\Graph'',\Data'')$ be the $\Aa_F$ configuration reached at that
  point. As $\delta$ is fireable from $(s_{\ell-1},w_{\ell-1})$ in $F$,
  since $(\Graph, \Data)\rhd(s_{\ell-1},w_{\ell-1})$, and as the
  content of the queue has not been modified since then, the head of
  \texttt{q} is necessarily an $n$ block in $\Graph''$. Moreover,
  $\Data''(\mathtt{head})=n$ and
  $\Data''(\mathtt{state})=s_\ell$. Thus, we let the scheduler dequeue
  this $n$ block, and we let the task running it execute line 10
  (where the condition of the \texttt{if} is not satisfied), then line
  11. Clearly, the resulting configuration encodes $(s_\ell, w_\ell)$.
  \medskip

  Now, let $\rho=(\Graph_0,\Data_0)a_0(\Graph_1,\Data_1)a_1\cdots a_{n-1}
  (\Graph_n,\Data_n)$ be a run of $\Aa_F$ s.t. there is $(s,w)$ with
  $(\Graph_n,\Data_n)\rhd(s,w)$, and let us build, by induction on the
  length of this run, a run $(s^0_F,w_0)(s_1,w_1)\cdots(s_k,w_k)$ a
  run of $F$ s.t. $(s_k,w_k)=(s,w)$.

  Let
  $K=|\set{(\Graph_i,\Data_i)}{\Parikh(\Graph_i)(\mystate{12}{m})=1\textrm{
      for }m\in M\cup\{\e\}}|$, i.e., $K$ is the number of times an
  $m$ block reaches line $12$ along $\rho$. Let us consider the
  increasing monotonic function
  $\rho:\{1,\ldots,K\}\rightarrow\{0,\ldots,n\}$ s.t. for all $0\leq
  i\leq n$: there exists $m\in M\cup\{\e\}$
  s.t. $\Parikh(\Graph_i)(\mystate{12}{m})=1$ iff there is $1\leq
  j\leq K$ s.t. $\Graph_i=\rho(j)$, that is, $\rho(i)$ is the index,
  in $\rho$ of the $i$th time a configuration is reached where an
  $\mathtt{m}$ block is at line 12. Clearly, by definition of $\rhd$
  only the $(\Graph_{\rho(j)},\Data_{\rho(j)})$ configurations (for
  $1\leq j\leq K$) can encode a configuration of $F$, as no $m$ block
  is at line 12 in the other configurations of $\rho$. So, it is
  sufficient to show that all those
  $(\Graph_{\rho(j)},\Data_{\rho(j)})$ configurations encode a
  reachable configuration of $F$. We proceed by induction on $j$, and
  show that: for all $1\leq j\leq K$:
  $(\Graph_{\rho(j)},\Data_{\rho(j)})$ encodes a reachable
  configuration of $F$ and $\Graph_{\rho(j)}$ contains exactly one $m$
  task (for $m\in M\cup\{e\}$), that has been dequeued from
  \texttt{q}.

  \textbf{Base case $j=0$:} Observe that the subrun
  $(\Graph_0,\Data_0)a_0\cdots
  a_{\rho(1)-1}(\Graph_{\rho(1)},\Data_{\rho(1)})$ is necessarily an
  initialization phase where \texttt{main} sets \texttt{state} to
  $s_F^0$, \texttt{head} to $\e$, dispatches an $\e$ block, and
  reaches line 8, where it will stay forever. Then, the scheduler
  dequeues the $\e$ block, which empties the queue. The $\e$ task then
  traverses line 10 (as \texttt{head}$=\e$) and 11 and reaches line
  12. So, clearly $(\Graph_{\rho(0)},\Data_{\rho(0)})\rhd (s^0_F,\e)$
  and contains exactly one $m$ task (for $m\in M\cup\{e\}$), that has
  been dequeued from \texttt{q}.

  \textbf{Inductive case $j=\ell$:} Let us assume that
  $(\Graph_{\rho(\ell-1)},\Data_{\rho(\ell-1)})$ encodes a reachable
  configuration $(s_{\ell-1}, w_{\ell-1})$ of $F$. We consider several
  cases. If $(\Graph_{\rho(\ell-1)},\Data_{\rho(\ell-1)}) =
  (\Graph_{\rho(\ell)},\Data_{\rho(\ell)})$ we are done. Otherwise, we
  have necessarily performed one iteration (possibly interrupted at
  line 12, 14 or 19) of the while loop at line 11 between
  $(\Graph_{\rho(\ell-1)},\Data_{\rho(\ell-1)})$ and
  $(\Graph_{\rho(\ell)},\Data_{\rho(\ell)})$, as, by induction
  hypothesis, $\Graph_{\rho(\ell-1)}$ contains exactly one $m$ task
  (with $m\in M\cup\{\e\}$) that blocks \texttt{q}, and \texttt{main}
  can only loop at line 8, which does not modify the current
  configuration. Then, observe that the conditions of the \texttt{if}
  at lines 12 and 14 were necessarily false during the
  iteration. Otherwise, $m$ would have reached line 21, from which it
  cannot escape. From that point, no configuration is reachable where
  an $\mathtt{m}$ block is at line 12 , and $(\Graph_{\rho(ell)},
  \Data_{\rho(\ell)})$ cannot exist.  Thus, we consider three cases:
  \begin{itemize}
  \item If we have entered the \texttt{if} at line 16 during the
    iteration, then a transition of the form $(s,!n,s')$ has been
    guessed, with $\texttt{state}=s$ and a dispatch of $n$ has been
    performed into $q$. As
    $(\Graph_{\rho(\ell-1)},\Data_{\rho(\ell-1)})\rhd (s_{\ell-1},
    w_{\ell-1})$ by induction hypothesis, $s_{\ell-1}=s$, and thus
    $(s,!n,s')$ is fireable from $(s_{\ell-1}, w_{\ell-1})$ and reaches
    $(s',n\cdot w_{\ell-1})$. Clearly, this configuration is encoded
    by $(\Graph_{\rho(\ell)},\Data_{\rho(\ell)})$.
  \item If we have entered the \texttt{else if} at line 17 during the
    iteration, then a transition of the form $(s,?n,s')$ has been
    guessed, with $\texttt{state}=s$, \texttt{head} has been set to
    $n$, the current $m$ block has been terminated, a new block $m'$
    has been dequeued by the scheduler (as there is necessarily a
    running \texttt{m} block in $\Graph_{\rho(\ell)}$). Moreover
    $m'=n$, because $m'$ has to be at line 12 in
    $\Graph_{\rho(\ell)}$, so the test of line 10 had to be false to
    allow $m'$ to reach line 12. As
    $(\Graph_{\rho(\ell-1)},\Data_{\rho(\ell-1)})\rhd (s_{\ell-1},
    w_{\ell-1})$ by induction hypothesis, $s_{\ell-1}=s$. As a dequeue
    of a block $m'=n$ has been performed, $w_{\ell-1}$ is of the form
    $w\cdot n$. Thus, $(s,?m,s')$ is fireable from $(s_{\ell-1},
    w_{\ell-1})$ and reaches $(s',w)$. Clearly, this configuration is
    encoded by $(\Graph_{\rho(\ell)},\Data_{\rho(\ell)})$.
  \item Finally, if neither the \texttt{if} nor the \texttt{else if}
    have been entered during the iteration, then a transition of the
    form $(s,\e,s')$ has been guessed, with $\texttt{state}=s$. As
    $(\Graph_{\rho(\ell-1)},\Data_{\rho(\ell-1)})\rhd (s_{\ell-1},
    w_{\ell-1})$ by induction hypothesis, $s_{\ell-1}=s$, and thus
    $(s,\e,s')$ is fireable from $(s_{\ell-1}, w_{\ell-1})$ and reaches
    $(s',w_{\ell-1})$. Clearly, this configuration is encoded by
    $(\Graph_{\rho(\ell)},\Data_{\rho(\ell)})$.\qed
  \end{itemize}
\end{proof}

We can now prove Theorem~\ref{the:async-seri-undec}:
\begin{proof}
  Let $F$ be a FIFO system, with set of messages $M$ and associated
  serial asynchronous \qdas $\Aa_F$ and let $c$ be a control location
  of $F$. For all $m\in M\cup\{\e\}$, let $f_m$ be the Parikh image
  s.t. $f_m(\mystate{21}{{\tt main}})=1$ and $f_m(s)=0$ for all $s\neq
  \mystate{21}{{\tt main}}$. Remark that there are only finitely many
  such $f_m$. Then, we show that $c$ is reachable in $F$ iff there
  exists $m\in M\cup\{\e\}$ s.t.  $f_m$ is Parikh-coverable in
  $\Aa_F$.

  Assume $c$ is reachable in $F$, and let $(c,w)$ be a configuration
  in $\Reach(F)$. Without loss of generality, assume $c$ is reachable
  by run that visits $c$ only once. By
  Lemma~\ref{lem:from-fifo-to-qdas}, there is
  $(\Graph,\Data)\in\Reach(\Aa_F)$
  s.t. $(\Graph,\Data)\rhd(c,w)$. Hence, in $(\Graph,\Data)$, there is
  a task running an $m$ block (for $m\in M\cup\{\e\}$) that is at line
  12, and $\Data(\mathtt{state})=c$. Thus, $m$ can execute one step
  and reach line 21, so $f_m$ is Parikh coverable in $\Aa_F$.

  For the reverse direction, assume there is $m\in M\cup\{\e\}$ that
  is Parikh-coverable in $\Aa_F$. Hence, there is
  $(\Graph,\Data)\in\Reach(\Aa_F)$ where a task running block $m$ is
  at line 21. The only way for that block to reach line 21 is from
  line 12, with a valuation $\Data'$
  s.t. $\Data'(\mathtt{state})=c$. Thus, there is, in $\Reach(\Aa_F)$
  a configuration $(\Graph',\Data')$ with $\Data'(\mathtt{state})=c$,
  a task running an $m$ block at line 12, and necessarily
  \texttt{main} at line 8 (otherwise, only \texttt{main} would be
  running). Hence, $(\Graph',\Data')$ is a reachable configuration of
  $\Aa_F$ s.t. $(\Graph',\Data')\rhd (c,w)$ for some queue content
  $w$. Thus, by Lemma~\ref{lem:from-fifo-to-qdas},
  $(c,w)\in\Reach(F)$, and $c$ is reachable in $F$.

  We have thus reduced the control location reachability problem of
  FIFO systems to the Parikh coverability problem of serial
  asynchronous \qdas (using only one serial queue). The former is
  undecidable. Hence the theorem.\qed
\end{proof}

\subsection{Concurrent \qdas}

We reduce the reachability problem of two counter
systems. Let us give the intuition of the construction. For each
$\Pp$, we construct a \qdas $\Aa_\Pp$ s.t. all reachable \ctg in
$\Aa_\Pp$ encode configurations of $\Pp$ and are of the form depicted
in Fig.\,\ref{fig:sim_counter}. That is, (after an initialization
phase), there are always three tasks that are unblocked: a \main task
to simulate $\Pp$'s control structure, and, for each $i=\{1,2\}$,
either a task $eins(i)$ or a task $null(i)$. If the task $null(i)$ is
unblocked, then counter $i$ is zero in the current configuration of
$\Pp$. Otherwise, the current valuation of counter $i$ is encoded by
the number of $eins(i)$ tasks in the \ctg. Remark that, as in the case
of synchronous \qdas, the parts of the \ctg that encode each counter
behave as pushdown stacks. Finally, the control location of $\Pp$ is recorded in
global variable \texttt{state}.

\begin{figure}[t!]
  \centering
  \begin{minipage}[t]{.5\linewidth}
      \begin{lstlisting}
global state
global $\ell_1^1$, $\ell_2^1$, $x^1$ // rdvz channel 1
global $\ell_1^2$, $\ell_2^2$, $x^2$ // rdvz channel 2
global c_queue q

def main():
  foreach i in {1,2}:
    dispatch_a(q, null(i))
    i?ack
  state := $x^0$

  while(true):
    select $(s,a,s')\in\Delta_\Pp$ where state=$s$

    if $a=\incr(1)$ :
      1!$\incr$
      1?$\ack$
      state:=$s'$

    \\ other actions analogous
    ...
      \end{lstlisting}
    \end{minipage}
    \begin{minipage}[t]{.5\linewidth}
\begin{tikzpicture}[baseline=(base),
  zstd/.style={state,font=\tiny}]
  \draw (0,0.3) coordinate (base);
  \draw (1.5,0) node[zstd] (1) {{1}};
  \draw (1)+(-0.5,0) node (s) {}
      (s) edge[->] (1);
  \draw (3,0) node[zstd] (2) {2};
  \draw (4,0) node[zstd] (f) {};
  \draw[fill=black] (f) circle (3pt);
  \draw (0,-1) node[zstd] (3) {3};
  \draw (1.5,-1) node[zstd] (4) {4};
  \draw (3,-1) node[zstd] (5) {5};
  \draw (3) edge[->,bend left=11] node[above,font=\tiny]{$i!\ack$} (4);
  \draw (4) edge[->,bend left=11] node[below,font=\tiny]{$i?\zerotest$} (3);
  \draw (1) edge[->] node[left,font=\tiny,pos=0.3]{$i!\ack$} (4);
  \draw (4) edge[->] node[below,font=\tiny]{$i?\incr$}(5);
  \draw (5) edge[->] node[right,font=\tiny]{$\disps(q,eins(i))$} (2);
  \draw (2) edge[->] node[left,pos=0.3,font=\tiny]{$i!\ack$\ }(4);

  \draw (0,0) node {$null(i)$:};
\end{tikzpicture}

\vspace{2ex}
\begin{tikzpicture}
  [zstd/.style={state,font=\tiny}]
  \draw (1.5,0) node[zstd] (1) {1};
  \draw (1)+(-0.5,0) node (s) {}
      (s) edge[->] (1);
  \draw (3,0) node[zstd] (2) {2};
  \draw (0,-1) node[zstd] (3) {};
  \draw[fill=black] (3) circle (3pt);
  \draw (1.5,-1) node[zstd] (4) {4};
  \draw (3,-1) node[zstd] (5) {5};
  \draw (4) edge[->] node[below,font=\tiny]{$i?\decr$} (3);
  \draw (1) edge[->] node[left,font=\tiny,pos=0.3]{$i!\ack$} (4);
  \draw (4) edge[->] node[below,font=\tiny]{$i?\incr$}(5);
  \draw (5) edge[->] node[right,font=\tiny]{$\disps(q,eins(i))$} (2);
  \draw (2) edge[->] node[left,pos=0.3,font=\tiny]{$i!\ack$\ }(4);

  \draw (0,0) node {$eins(i)$:};
\end{tikzpicture}

\vspace{2ex}
\begin{tikzpicture}
  [zstd/.style={state,font=\tiny},anchor=west]

  \draw (-0.5,1.2) node[anchor=west] (dummy){\ctg in $\Reach(\Aa_\Pp)$:};
  \draw (0,0.7) node[zstd] (0) {\main};
  \draw (0,0) node[zstd] (00) {null(1)};
  \draw (00)+(1,0) node[zstd] (01) {eins(1)};
  \draw (01)+(1,0) node[font=\tiny,anchor=west] (02) {\dots};
  \draw (02)+(0.5,0) node[zstd] (03) {eins(1)};
  \draw (0,-0.8) node[zstd] (11) {null(2)};
  \draw (11)+(1,0) node[zstd] (12) {eins(2)};
  \draw (12)+(1,0) node[font=\tiny,anchor=west] (13) {\dots};
  \draw (13)+(1.2,0) node[zstd] (14) {eins(2)};
  \draw[dashed] (00) edge[->] (01)
    (01) edge[->] (02)
    (02) edge[->] (03)
    (11) edge[->] (12)
    (12) edge[->] (13)
    (13) edge[->] (14);

    \begin{pgfonlayer}{background}
      \draw node[fit=(00) (03),fill=gray!20] (c1) {};
      \draw node[fit=(11) (14),fill=gray!20] (c2){};
      \draw (c1.north east)
        node[anchor=south east,inner sep=0pt,font=\tiny]
        {counter 1};
      \draw (c2.north east)
        node[anchor=south east,inner sep=0pt,font=\tiny]
        {counter 2};
    \end{pgfonlayer}
\end{tikzpicture}
    \end{minipage}

    \caption{From a two counter system  to a \qdas: \main and
    $null(i),eins(i)$ for $i=1,2$ \label{fig:sim_counter}}
    \vspace{-4ex}
\end{figure}

The actual operations on the counters will be simulated by the
$eins(i)$ and $null(i)$ running tasks. As \main simulates the control
structure, we need to synchronize
\main with those $eins(i)$ and $null(i)$ tasks.  Let us explain
intuitively how we can achieve \emph{rendezvous} synchronization
between running tasks using global variables of \qdas. Consider a
\qdas with three global variables $\ell_1$, $\ell_2$ ranging over
Boolean and $X$ over a finite set of `messages' $M$.
Let $\gamma_1$ and $\gamma_2$ be two blocks whose \lts
are:\\
$\gamma_1$: \begin{tikzpicture}[baseline=-0.5ex] \foreach \x in
  {0,1,2,3,4,5} \draw (1.5*\x,0) node[circle,draw,inner sep=4pt] (\x)
  {}; \draw (0) node[font=\tiny]{$s_0$}; \draw (5)
  node[font=\tiny]{$s_5$}; \draw (0) edge[->]
  node[above,font=\tiny]{$\ell_1 = 1$} (1); \draw (1) edge[->]
  node[above,font=\tiny]{$x\gets m$} (2); \draw (2) edge[->]
  node[above,font=\tiny]{$\ell_2\gets 1$} (3); \draw (3) edge[->]
  node[above,font=\tiny]{$\ell_1=0$} (4); \draw (4) edge[->]
  node[above,font=\tiny]{$\ell_2\gets 0$} (5); \draw (0)+(-0.7,0) node
  (z) {} (z) edge[->](0);
\end{tikzpicture}\hfill (for $m\in M$)\\
$\gamma_2$: \begin{tikzpicture}[baseline=-0.5ex]
  \foreach \x in {0,1,2,3,4,5}
  \draw (1.5*\x,0) node[circle,draw,inner sep=4pt] (\x) {};
  \draw (0) node[font=\tiny]{$s_0'$};
  \draw (5) node[font=\tiny]{$s_5'$};
  \draw (0) edge[->] node[above,font=\tiny]{$\ell_1\gets 1$} (1);
  \draw (1) edge[->] node[above,font=\tiny]{$\ell_2=1$} (2);
  \draw (2) edge[->] node[above,font=\tiny]{$x=m$} (3);
  \draw (3) edge[->] node[above,font=\tiny]{$\ell_1\gets 0$} (4);
  \draw (4) edge[->] node[above,font=\tiny]{$\ell_2= 0$} (5);
  \draw (0)+(-0.7,0) node (z) {} (z) edge[->](0);
\end{tikzpicture}\\
Assume a configuration $c$ of the \qdas where $\ell_1=\ell_2=0$ and
where two distinct tasks are running $\gamma_1$ and $\gamma_2$, are
unblocked, and are in $s_0$ and $s_0'$ respectively.  Assume that no
other task can access $\ell_1$, $\ell_2$ and $m$. It is easy to check
that, from $c$, there is only one possible interleaving of the transitions
of$\gamma_1$ and $\gamma_2$. So if $\gamma_2$ reaches $s_5'$ from $c$,
then $\gamma_1$ must have reached $s_5$, and the $x=m$ test in
$\gamma_1$ has been fired \emph{after} the $x\gets m$ assignment in
$\gamma_2$. This achieves rendezvous synchronisation between
$\gamma_1$ and $\gamma_2$, with the passing of message $m$. This can
easily be extended to rendezvous via different ``channels'', by adding
extra global variables. So, we extend the syntax of
\qdas by allowing transitions of the form $(s_0,c!m,s_5)$ and
$(s_0',c?m,s_5')$ (for $m\in M$) to denote respectively a send and a
receive of message $m$ on a rendezvous channel $c$.

We rely on this mechanism to let \main send operations to be performed
on the counters to the $null(i)$ and $eins(i)$ running tasks. More
precisely, for a \twocs $\Pp=\tuple{X,x^0,\Sigma_\Pp,\Delta_\Pp}$, we
build the \qdas $\Aa_\Pp=\QDAS$ where $\CQID=\{q\}$,
$\Gamma=(\{null,eins\}\times\{1,2\})\cup\{\main\}$,
$\Xx=\{\ell_1^1,\ell_2^1,x^1,\ell_1^2,\ell_2^2,x^2\}$ where $x^1,x^2$
range over the domain $\{\incr,\decr,\zerotest,\ack\}$, and the
transition systems are given in Fig.\,\ref{fig:sim_counter}. The
variables $\Xx$ encode two channels that we call $1$ and $2$ in the
pseudo code of Fig.~\ref{fig:sim_counter}. The \main task runs an
infinite \texttt{while} loop (line 12 onwards) that consists in
guessing a transition $(s,a,s')$ of $F$ and synchronising, via
\emph{rendezvous} on the channels $1$ and $2$, with the relevant
$null$ or $eins$ unblocked task, to let it execute the operation on
the counter. When a $null(i)$ or $eins(i)$ receives an $\incr$
message, it performs an asynchronous dispatch of $eins(i)$ into $q$ to
increment counter $i$, and acknowledges the operation to \main, thanks
to message $\ack$. When an $eins$ block receives a $\decr$ message, it
terminates, which decrements the counter. $null$ blocks cannot receive
$\decr$ messages, so, if \main requests a $\decr$ operation when the
counter is zero, \main gets blocked. This means that the guessed
transition was not fireable in the currently simulated \twocs configuration, and ends the
simulation. Finally, only $null$ blocks can receive and acknowledge
$\zerotest$ messages, so, again, \main is blocked after sending
$\zerotest$ to a non-zero counter.
Note that we need both \emph{asynchronous} calls to start two counters
in parallel, and \emph{synchronous} calls to encode the counter
values. The result of Theorem\,\ref{thm:concqdasundec} follows
directly from:

\begin{restatable}{proposition}{propsimpdsrdvzqdas}
  \label{prop:sim_pds_rdvz_qdas}
  Given a \twocs,
  then we can reduce its reachability question
  to the Parikh coverability question for a concurrent \qdas that demands
  both synchronous and asynchronous dispatch actions.
\end{restatable}

As discussed before, we can separate each $\Graph$ for
$(\Graph,\Data)\in\Reach(\Aa_\Cc)$ into three components, one
consisting only of a vertex $v_0$ with $\lambda(v_0)=\main$
and two paths $v_1 v_2 \dots v_k$ and $v_1'v_2'\dots v_l'$ which
we will call $counter1$ and $counter2$ in the following.

As before, we define a relation between configurations of the \twocs $\Cc$ and
the \qdas $\Aa_{Cc}$. For $c=(\Graph,\Data)\in\Reach(\Aa_\Cc)$ and
$y=(x,k,l)\in\Reach(\Cc)\subseteq X\times \NN \times \NN$ we
write $c\triangleright y$ if $\Data(state)=x$, $\veve{counter1}=k$,
and $\veve{counter2}=l$.

The rendezvous assures a unique interleaving of actions of \main, $null(1)$,
and $null(2)$ until \main reaches line $12$. Let us in the following consider
the reached configuration $c^0=(\Graph^0,\Data^0)$ with $\Data^0(state)=x^0$,
$\Data^0(\ell_0)=\Data^0(\ell_1)=0$ and
$\Graph^0$ with\\
\begin{tikzpicture}
  [zstd/.style={state,font=\small,inner sep=1pt,minimum size=15pt},
  zstd2/.style={zstd,rectangle},
  lab/.style={font=\tiny,inner sep=1pt},
  anchor=west]
    \draw  node[zstd] (13) {$\main$};
    \draw (13.south east) node[lab,anchor=north west] {$q$};
    \draw (13.north east) node[lab,anchor=south west] {$s_{12}$};

    \draw (1.5,0) node[zstd] (13) {$null(1)$};
    \draw (13.south east) node[lab,anchor=north west] {$q$};
    \draw (13.north east) node[lab,anchor=south west] {$4$};

    \draw (3.5,0) node[zstd] (13) {$null(2)$};
    \draw (13.south east) node[lab,anchor=north west] {$q$};
    \draw (13.north east) node[lab,anchor=south west] {$4$};
\end{tikzpicture}\\
(where $s_{12}$ is the state of \main in line $12$) as ``initial'' configuration of the \qdas.

Note that $counter1$ and $counter2$ are independent, i.e., they do not
synchronize except via $\main$. Further, there is no more than one task active
in $counter1$ and $counter2$. The unique tasks $zero(1)$ and $zero(2)$
never terminate. The rendezvous synchronization assures that there is only
\emph{one} possible interleaving between the \main task and the currently running
tasks in $counter1$ and $counter2$:
\begin{myitemize}
  \item \main does loops of the form\\
\begin{tikzpicture}[baseline=-0.5ex]
  \foreach \x in {0,1,2,3,4}
    \draw (1.8*\x,0) node[circle,draw,inner sep=4pt] (\x) {};

  \draw (0) node[font=\tiny]{$s_0$};
  \draw (4) node[font=\tiny]{$s_0$};
  \draw (0) edge[->] node[above,font=\tiny]{$state = s$} (1);
  \draw (1) edge[->] node[above,font=\tiny]{$i!a$} (2);
  \draw (2) edge[->] node[above,font=\tiny]{$i?ack$} (3);
  \draw (3) edge[->] node[above,font=\tiny]{$state\leftarrow s'$} (4);
\end{tikzpicture}\hfill (for $i\in \{1,2\},a\in\Sigma_\Cc$)\\

  \item which leads to the following interleaving of actions of \main
    with actions of the $i$-th counter component.\\
\begin{tikzpicture}
  \foreach \x in {0,1,2,3,4,5,6,7}
    \draw (1.7*\x,0) node[circle,draw,inner sep=4pt] (\x) {};

  \draw (0) edge[->] node[above,font=\tiny]{$state = s$} (1);
  \draw (1) edge[->] node[above,font=\tiny]{$i!a$} (2);
  \draw (2) edge[->] node[above,font=\tiny]{$i?a$} (3);
  \draw (3) edge[->,dotted] node[above,font=\tiny]{$c_i(a)$} (4);
  \draw (4) edge[->] node[above,font=\tiny]{$i!ack$} (5);
  \draw (5) edge[->] node[above,font=\tiny]{$i?ack$} (6);
  \draw (6) edge[->] node[above,font=\tiny]{$state\leftarrow s'$} (7);
\end{tikzpicture}\\

where
\begin{tikzpicture}
  \draw node[circle,draw,inner sep=2pt] (a) {};
  \draw +(1,0) node[circle,draw,inner sep=2pt] (b) {};
  \draw (a) edge[->,dotted] node[above,font=\tiny]{$c_i(a)$} (b);
\end{tikzpicture}
translates the sent action $a$ to a meta-action $c_i(a)$ of the $i$-th counter
as follows:\\
\begin{myitemize}
  \item an action $\incr$ is mapped to the action
    $\disps(q,eins(i))$ and the activation of the dispatched task
  \item an action $\decr$ is mapped to the termination of
    the current task; which is only possible if the current task
    is a block $eins(i)$
  \item the test for empty stack is mapped to an epsilon action; this action
    is only possible in $null(i)$.
\end{myitemize}
\end{myitemize}
Note that if $c_i(a)$ is not possible, then there will be no acknowledgement,
hence $\Aa_\Cc$ blocks.

Thus we can cut a run of $\Aa_\Cc$ into (an initial phase and) a sequence of
phases of the above form that will be abbreviated $trans(s,a,s')$ in the
following.

\begin{lemma}
  Let $\Cc$ be a \twocs and $\Aa_{\Cc}$
  the associated \qdas, if  $y\in\Reach(\Cc)$ then there exists
  $c\in\Reach(\Aa_{\Cc})$ such
  that $c\triangleright y$. Further, if $c=(\Graph,\Data)\in\Reach(\Aa_\Cc)$
  where $\Data$ valuates $\Data(\ell_0)=\Data(\ell_1)=0$,
  then there exists $y\in\Reach(\Cc)$ with $c\triangleright y$.
\end{lemma}

\begin{proof}

  Given a run $x_0 \delta_1 x_1 \delta_2 \dots \delta_k x_k$ of
  $\Cc$, then there exists a run of $\Aa_\Cc$ that can be cut into
  phases $t_1,\dots,t_k$ where $t_i=trans(s_{i-1},a_i,s_i)$ where
  $\delta_i=(s_{i-1},a_i,s_i)$ for $1\leq i \leq k$.
    Obviously $c^0\triangleright x^0$ and
  $\Data^0(\ell_0)=\Data^0(\ell_1)=0$.
  Hence, the reverse direction follows by a straightforward inductive
  argument.
      \end{proof}

\end{document}